\documentclass[nofootinbib]{revtex4-2}
\usepackage{amsmath,amsfonts,amsthm,amssymb,bbm}
\usepackage{graphicx,color,dsfont}
\usepackage{enumitem}
\usepackage{fourier}
\usepackage{url}
\newtheorem{theorem}{Theorem}
\newtheorem{lemma}{Lemma}
\newtheorem{proposition}{Proposition}

\newtheorem{definition}{Definition}
\newtheorem{corollary}{Corollary}

\newtheorem{claim}{Claim}

 \theoremstyle{definition}
 
 \theoremstyle{remark}

 \numberwithin{equation}{section}

\newcommand{\vertiii}[1]{{\left\vert\kern-0.25ex\left\vert\kern-0.25ex\left\vert #1
    \right\vert\kern-0.25ex\right\vert\kern-0.25ex\right\vert}}

\newcommand{\f}[2]{\frac{#1}{#2}}

\newcommand{\cl}{{\mathcal L}}

\newcommand{\al}{\alpha}
\newcommand{\be}{\beta}

\newcommand{\ga}{\gamma}

\newcommand{\de}{\delta}

\newcommand{\De}{\Delta}
\newcommand{\ve}{\varepsilon}

\newcommand{\la}{\lambda}

\newcommand{\si}{\sigma}

\newcommand{\vp}{\varphi}

\newcommand{\om}{\omega}

\newcommand{\rd}{{\mathbf R}^d}

\newcommand{\rone}{\mathbf R}

\newcommand{\dpr}[2]{\langle #1,#2 \rangle}

\renewcommand{\sd}{\mathbf S^{d-1}}

\newcommand{\eps}{\epsilon}

\newcommand{\p}{\partial}

\newcommand{\beq}{\begin{equation}}
\newcommand{\eeq}{\end{equation}}
\newcommand{\beqna}{\begin{eqnarray*}}
\newcommand{\eeqna}{\end{eqnarray*}}
\newcommand{\beqn}{\begin{equation*}}
\newcommand{\eeqn}{\end{equation*}}
\newcommand{\bp}{\begin{proof}}
\newcommand{\ep}{\end{proof}}
\newcommand{\bprop}{\begin{proposition}}
\newcommand{\eprop}{\end{proposition}}
\newcommand{\bt}{\begin{theorem}}
\newcommand{\et}{\end{theorem}}
\newcommand{\bex}{\begin{Example}}
\newcommand{\eex}{\end{Example}}
\newcommand{\bc}{\begin{corollary}}
\newcommand{\ec}{\end{corollary}}
\newcommand{\bcl}{\begin{claim}}
\newcommand{\ecl}{\end{claim}}
\newcommand{\bl}{\begin{lemma}}
\newcommand{\el}{\end{lemma}}

\newcommand{\cj}{{\mathcal J}}

\begin{document}

\title{Mixed dispersion nonlinear Schr\"odinger equation
in higher dimensions: theoretical analysis and numerical computations}

\author{Atanas Stefanov}
\affiliation{ Department of Mathematics, University of Alabama - Birmingham,
	1402 10th Avenue South
	Birmingham AL 35294, USA
}
\email{stefanov@uab.edu}

\author{Georgios A. Tsolias}
\affiliation{ Department of Mathematics and Statistics,
University of Massachusetts, Amherst
Amherst, MA 01003, USA}
\email{gtsolias@umass.edu}

\author{Jes\'us Cuevas-Maraver}
\affiliation{Grupo de F\'{\i}sica No Lineal, Departamento de F\'{\i}sica Aplicada I,
Universidad de Sevilla. Escuela Polit\'{e}cnica Superior, C/ Virgen de Africa, 7, 41011-Sevilla, Spain \\
Instituto de Matem\'{a}ticas de la Universidad de Sevilla (IMUS). Edificio
Celestino Mutis. Avda. Reina Mercedes s/n, 41012-Sevilla, Spain}
\email{jcuevas@us.es}

\author{Panayotis G. Kevrekidis}
\affiliation{Department of Mathematics and Statistics,
University of Massachusetts, Amherst
Amherst, MA 01003-9305, USA}
\email{kevrekid@umass.edu}

\date{\today}

\begin{abstract}
 In the present work we provide a characterization of the ground states
 of a higher-dimensional quadratic-quartic model of the nonlinear
 Schr{\"o}dinger class with a combination of a focusing biharmonic operator
  with either an isotropic or an anisotropic
 defocusing Laplacian operator (at the linear level)
 and power-law nonlinearity.
 Examining principally the prototypical
 example of dimension $d=2$,
 we find that instability
arises beyond a certain threshold coefficient of the Laplacian between
the cubic and quintic cases, while all solutions are stable for powers below the cubic. Above the quintic, and up to a critical nonlinearity
exponent $p$, there exists a progressively narrowing range of stable
frequencies. Finally, above the critical $p$ all solutions are unstable.
The picture is rather similar in the anisotropic case, with the difference
that even before the cubic case, the numerical computations suggest
an interval of unstable frequencies. Our analysis generalizes the
relevant observations for arbitrary combinations of Laplacian
prefactor $b$ and nonlinearity power $p$.
\end{abstract}

\maketitle

   \section{Introduction}

   The nonlinear Schr{\"o}dinger equation~\cite{sulem,ablowitz1,mjarecent,siambook} is a well-established model
   characterizing a wide range of physical settings extending from
   deep water waves~\cite{mjarecent}  to electromagnetic wave evolution in optical
   fibers~\cite{hasegawa,kivshar}, and from dilute gases of atomic condensates~\cite{siambook,stringari,pethick} to
   waves in plasmas~\cite{plasmas,ir}. However, in recent years,
   biharmonic dispersive-wave models have been gaining considerable
   traction due to their own emergence in a variety of applications.
   Arguably, one of the most notable examples thereof is due to the
   experimental realization of dispersion engineering in
   laboratory optical systems~\cite{pqs} that enabled quartic dispersion
   and, through the competition of that with nonlinearity, the formation
   of the so-called pure-quartic solitons (PQS)~\cite{pqs}.
   This type of dispersion engineering is also responsible for
   the realization of the so-called pure-quartic soliton laser~\cite{pqs3}.
   In fact, even pure dispersion of higher orders has been
   explored, e.g., in~\cite{RungePRR2021}, giving rise to
   considerations of conveniently programmable dispersions
   in fiber lasers~\cite{RungePRR2021}.
   { As explained in further detail, e.g., in~\cite{pqs},
   the prototypical physical system, on account of which such models have
   recently been intensely explored, involves dispersion-engineered photonic crystal waveguides.
   In such optical settings, the ability to suitably
   manipulate dispersion in competition with effects of
   self-phase modulation is the principal reason for the
   ability to generate, as well as experimentally observe such
   solitary waves.}
   Motivated by these developments, numerous mathematical works
   have explored the existence and stability of solitary waves
   in such settings~\cite{beam_demirkaya,beam1,atanas}.   The accessibility of a wide range
   of dispersion profiles has revived considerations
   of different co-existing types of dispersion, e.g.,
   in the form of a quadratic and a quartic term, as e.g.
   in the work of~\cite{pqs2}. The latter setting has been
   previously explored, e.g., in the classic works
   of~\cite{KarpmanPLA1994,KarpmanPRE96}. Similar competitions
   between the quadratic and quartic dispersion have also
   been recently considered in the wave equation setting, giving
   rise to different bound states and intriguing dynamical
   phenomena~\cite{TsoliasJPA2021}.

   It is indeed this topic of competing Laplacian and biharmonic terms
   that we revisit in the present work, especially with a view
   to higher-dimensional considerations and the interplay of the
   power (exponent) $p$ of the nonlinearity and the dimensionality $d$
   of the linear operator.
  More precisely, we shall consider the following mathematical models:
    \begin{eqnarray}
    \label{10}
    & & 	 i u_t +\De^2 u+b \De u   - |u|^{p-1} u=0, \ \ x\in \rd \\
    \label{12}
    & & 	 i u_t +\De^2 u+b \p_{x_1}^2u   - |u|^{p-1} u=0, \ \ x\in \rd
    \end{eqnarray}
   Our work will be about the study of solitary waves of such models and their
   stability properties. In fact, we consider standing waves in the form $u=e^{-i \om t} \Phi$, which results in the elliptic profile equations:
   \begin{eqnarray}
   \label{20}
   & & \De^2 \Phi +b \De \Phi+\om \Phi   - |\Phi|^{p-1} \Phi=0, \ \ x\in \rd \\
   \label{22}
   & & \De^2 \Phi +b \p_{x_1}^2  \Phi+\om \Phi   - |\Phi|^{p-1} \Phi=0, \ \ x\in \rd \
   \end{eqnarray}
   We will refer to the model \eqref{10} as the isotropic case, while the model \eqref{12} as the anisotropic case { (due
   to its different dispersion along the direction $x_1$)}.
   Next, we set up the linear stability framework for these models. Namely, taking $u=e^{- i \om t}(\Phi+v)$, plugging this in \eqref{10} (or \eqref{12} respectively) and ignoring the higher order terms (i.e. super-linear ones of the form $O(v^2)$), we obtain for $\vec{v}=(\Re v, \Im v)$,
   \begin{eqnarray}
   \label{219}
   \vec{v}_t &=& \cj \cl \vec{v}, \cj=\left(\begin{array}{cc}
   0 & -1 \\ 1 & 0
   \end{array}\right), \cl=\left(\begin{array}{cc}
   \cl_ +& 0 \\ 0 & \cl_-
   \end{array}\right) \\
   \label{221}
   \cl_+ &=& \De^2+b\De + \om - p |\Phi|^{p-1},  \\
   \label{222}
   \cl_- &=&  \De^2+ b \De +\om -  |\Phi|^{p-1}.
   \end{eqnarray}
   Similarly, the eigenvalue problem for the anisotropic model \eqref{12} is  also in the form \eqref{219}, with $\cl_\pm$ given by
   $$
   \left\{
   \begin{array}{ll}
   \cl_+ &= \De^2+ b \p_{x_1}^2 + \om - p |\phi|^{p-1},  \\
   \cl_- &= \De^2+ b \p_{x_1}^2 +\om -  |\phi|^{p-1}
   \end{array}
   \right.
   $$

   We now give a formal definition of spectral stability, which, in the context
   of the standing waves of the model
   of interest, is the central focus
   of the present work.
   \begin{definition}
   	\label{defi:10} We say that the corresponding standing wave solution $e^{- i \om t} \Phi$ is spectrally stable, provided the eigenvalue problem $\cj \cl v= \mu v$ does not have non-trivial solution $(v, \mu): v\in H^4(\rd), \mu: \Re \mu>0$.
   \end{definition}

 The closest in spirit work to the present
   one is that of~\cite{KarpmanPRE96}.  In it, however, the author considers a different model, namely
   \begin{equation}
   \label{karp:10}
    i u_t +\gamma \De^2 u+\De u   + |u|^{p-1} u=0, \ \ x\in \rd.
   \end{equation}
   The authors obtains a number of useful (and mostly rigorous) results for the standing waves for these models, especially in the regime\footnote{ Note that the case $\ga>0$ in \eqref{karp:10} is not as relevant physically as it does not support (bright) localized waves. It basically corresponds to the de-focusing case in the standard NLS framework.} $\gamma<0, |\ga|<<1$. Note however that this case, after some rescaling is equivalent to the case $b<0$ in \eqref{10}, whereas our main interest  is in the case $b>0$. The latter involves a competition
   (rather than a cooperation) of the
   linear contributions and, hence,
   represents a case of particular interest.

   In the present
   setting, we will examine systematically  the isotropic case,
   but also compare it with the anisotropic one whereby the
   Laplacian operator is replaced by a second partial derivative along
   only a single spatial direction. We will present theoretical results
   in both cases for the ground states of the system and their stability
   as a function of the nonlinearity power $p$ and the coefficient of
   the Laplacian (or of the one-dimensional second partial derivative)
   $b$. Our principal theorems are,
   accordingly, stated in the next
   Section. {
   We note in passing the broad appeal of the topic of spectral
   stability of solitary waves in a variety of related nonlinear dispersive-wave    models, as evidenced, e.g., in the works of~\cite{Yang2010,Wang2016,Gong2022}.}

   We corroborate our theoretical analysis with detailed numerical computations
   that illustrate systematically both the isotropic and the anisotropic
   case with $d=2$, as a function of $b$ and also as a function of $p$.
   Starting with the isotropic case,
   we find that up to the cubic case of $p=3$, the relevant ground
   states are generically stable, irrespectively of the value of $b$.
   Beyond $p=3$ and for $3<p<5$, a critical threshold of $b$ exists
   such that below the relevant threshold, the wave is spectrally stable,
   while above, it destabilizes. Further, above $p=5$ and below a critical
   $p$, the waves will only be spectrally stable for an interval of
   $b$'s, while upon crossing this critical threshold, a saddle-center
   bifurcation leads to the disappearance of all stable solutions of the
   isotropic setting. Interestingly, the anisotropic example bears
   numerous similarities with the above described isotropic case.
   The most notable difference that is worth highlighting is that even
   below $p=3$, the anisotropic case may bear instabilities for a narrow interval
   of $b$-values; more details will be shown in our numerical computations that
   follow.

\section{Mathematical Setup and main results}

   We start by noting that
   the problem of interest possesses
continuous spectrum, which effectively,
per Weyl's theorem~\cite{kato}, amounts to the spectrum
of the homogeneous background state: $\si(\cl_\pm)=\si(\De^2+b\De)=\mathrm{Range}[\xi\to |\xi|^4-b|\xi|^2+\om]=[\om-\f{b^2}{4}, +\infty)$. If we do not expect embedded eigenvalues in the essential spectrum\footnote{However,  there are fourth order differential operators with fast decaying potentials, who have embedded eigenvalues in its continuous spectrum. This is in sharp contrast with the second order operators, who may possess eigenvalues only at the edges of the continuous spectrum.}, and since by a direct inspection $\cl_-[\Phi]=0$, so $0\in \si_{p.p.}(\cl_-)$ (where $\si_{p.p.}$
denotes the pure point spectrum), then we can clearly conclude $\om\geq \f{b^2}{4}$, which
corresponds to the range of frequencies
of the standing wave that we will be
considering in what follows.

Our principal theme of study will
consist of the standing wave solutions of \eqref{10}, \eqref{12} (that is, the solutions of \eqref{20}, \eqref{22}). Our main interest is in the (spectral) stability of these waves. For future reference,  we introduce the associated  Hamiltonian functionals,
	\begin{eqnarray*}
		I(u) &=&  \left\{\f{1}{2}  \int_{\rd} |\De u|^2- \f{b}{2} \int_{\rd} |\nabla u|^2 - \f{1}{p+1} \int_{\rd} |u|^{p+1}\right\},  \\
		J(u) &=&  \left\{\f{1}{2}  \int_{\rd} |\De u|^2- \f{b}{2} \int_{\rd} |u_{x_1}|^2 - \f{1}{p+1} \int_{\rd} |u|^{p+1}\right\}
	\end{eqnarray*}
 and  the associated constrained  minimization problems
 \begin{eqnarray}
 \label{e:10}
& &  \left\{
 \begin{array}{l}
 I(u)\to \min \\
 \int_{\rd} |u(x)|^2 dx=\la.
 \end{array}
 \right., \\
 \label{e:11}
 & & \ \left\{
 \begin{array}{l}
 J(u)\to \min \\
 \int_{\rd} |u(x)|^2 dx=\la.
 \end{array}
 \right.
 \end{eqnarray}
 The solutions of these problems, if they exist, are referred to as normalized waves for the corresponding variational problems.

For the rest of the paper, we consider the case $b>0$ only.
We have the following result.
\begin{theorem}(The isotropic case)
	\label{theo:10}
	Let $d\geq 2$ and $b>0$. Then, there exists a unique $p_*(d) \in (1+\f{4}{d}, 1+\f{8}{d+1})$, so that
	\begin{itemize}
		\item For $1<p<p_*(d)$, the constrained minimization problem \eqref{e:10} has a solution for every $\la>0$. Moreover, such solutions satisfy the Euler-Lagrange equation
		\begin{equation}
		\label{47}
		\De^2 \Phi+b \De \Phi+\omega \Phi - |\Phi|^{p-1} \Phi=0, x\in \rd,
		\end{equation}
		for some $\omega=\omega(\la)>0$. In addition, all the functions $e^{- i \omega(\la) t}\Phi$ 		 are spectrally stable in the context of the isotropic NLS \eqref{10}.
		\item For $1+\f{8}{d}>p>p_*(d)$, there exists $\la_*(p,d,b)>0$, so that the problem \eqref{e:10} has solutions for all $\la>\la_*(p,d)$. These solutions are spectrally stable.
	\end{itemize}
\end{theorem}
Our numerical results suggest that
$3.2<p_*(2)<3.4$, with the relevant
value being in the vicinity of $p_*(2) \approx
3.3$, yet the subtle nature of the numerical
considerations near the limit only affords
us an approximate result in this context.

Next, we have the following result regarding the anisotropic case.
\begin{theorem}(The anisotropic case)
	\label{theo:12}
	Let $d\geq 2$ and $b>0$. Then,
	\begin{itemize}
		\item For $1<p<1+\f{4}{d}$, the constrained minimization problem \eqref{e:11} has a solution for every $\la>0$. Moreover, all of these solutions are spectrally stable.
		\item For $1+\f{8}{d}>p>1+\f{4}{d}$, there exists $\la_*(p,d,b)>0$, so that the problem \eqref{e:10} has solutions for all $\la>\la_*(p,d)$. These solutions are spectrally stable.
	\end{itemize}
\end{theorem}
{\bf Remark:} The statement of the Theorem \ref{theo:12} does not imply that {\it all} waves are spectrally stable, but rather only that the  minimizers of the constrained minimization problem \eqref{e:11} are guaranteed to be spectrally stable. In fact, in later sections, we numerically explore waves (i.e. functions satisfying \eqref{22}), which are not necessarily spectrally stable. Interestingly they happen to co-exist with stable constrained minimizers in that for a range of
$p$, there exist multiple solutions
corresponding to different
frequencies with some (2) of them
being stable and one unstable.
We now turn to the systematic construction
of the waves of interest.

\section{Construction  of the waves: Preliminaries}
We begin our considerations with an analysis of when the constrained minimization problem  \eqref{e:10} is well-posed. That is, whether the quantity $I[u]$ is bounded from below.
\subsection{Well-posedness of the constrained minimization problems}
To this end, introduce  the following functions
\begin{eqnarray*}
	m(\la) &=&   \inf\limits_{\|u\|_{L^2}^2=\la} \left\{\f{1}{2}  \int_{\rd} |\De u|^2- \f{b}{2} \int_{\rd} |\nabla u|^2 - \f{1}{p+1} \int_{\rd} |u|^{p+1}\right\} \\
	n(\la) &=&   \inf\limits_{\|u\|_{L^2}^2=\la} \left\{\f{1}{2}  \int_{\rd} |\De u|^2-
	\f{b}{2} \int_{\rd} |u_{x_1}|^2 - \f{1}{p+1} \int_{\rd} |u|^{p+1}\right\}
\end{eqnarray*}

It is not {\it a priori} clear that $m(\la), n(\la)$ is finite. We have the following result detailing that.
\begin{lemma}
	\label{le:10} Let $d\geq 1$. Then, for every $\la>0$ and $1<p<1+\f{8}{d}$, we have that $-\infty<m(\la)<0$.
	
	For $p>1+\f{8}{d}$, $m(\la)=-\infty$.
\end{lemma}
\begin{proof}
Assume that $1<p<1+\f{8}{d}$.
	By the Gagliardo-Nirenberg-Sobolev's inequalities
	\begin{eqnarray*}
	\|u\|_{L^{p+1}(\rd)}\leq C \|u\|_{\dot{H}^{d(\f{1}{2}-\f{1}{p+1})}}\leq
	C_{d,p} \|\De u\|_{L^2}^{\f{d}{2}(\f{1}{2}-\f{1}{p+1})} \|u\|_{L^2}^{1-\f{d}{2}(\f{1}{2}-\f{1}{p+1})}.
	\end{eqnarray*}
	Thus, for a function $u: \|u\|^2=\la$, we have
	$$
	\|u\|_{L^{p+1}(\rd)}^{p+1}\leq \la^{\f{1}{2}\left(p+1-\f{d}{4}(p-1)\right)} \|\De u\|_{L^2}^{\f{d(p-1)}{4}}=: C_\la \|\De u\|_{L^2}^{\f{d(p-1)}{4}}.
	$$
Since $	\f{d(p-1)}{4}<2$, we conclude by Young's inequality that for every $\de>0$,
$$
\|u\|_{L^{p+1}(\rd)}^{p+1}\leq  \left(\f{C_\la}{\de^{\f{d(p-1)}{8}}}\right)^{\f{8}{8-d(p-1)}} + \de \|\De u\|_{L^2}^2 \leq
D_{\la, \de,p} + \de \|\De u\|_{L^2}^2.
$$
	Trivially, $\|\nabla u\|^2\leq C_d \|\De u\| \|u\|\leq \de \|\De u\|^2+ \f{C_d^2 \la}{\de} $, so by setting $\de=\de_{\la,p,b}$ appropriately small, we obtain that for a function $u: \|u\|^2=\la$,
	\begin{equation}
	\label{112}
	I[u]\geq \f{1}{4} \|\De u\|_{L^2}^2 - C_{\la,p,b}.
	\end{equation}
	whence the function $m$ is bounded from below.

	Let $\phi$ be a test  function, $\|\phi\|_{L^2}^2=\la$. We take the scaling transformation $\phi_a=a^{d/2} \phi(a x)$, so $\|\phi_a\|_{L^2}^2=\la$. We have
	$$
	I[\phi_\eps]=a^4 \f{\|\De \phi\|^2}{2}-b a^2 \f{\|\nabla \phi\|^2}{2}
	- \f{a^{\f{d(p+1)}2-d}}{p+1}  \|\phi\|_{L^{p+1}}^{p+1}.
	$$
	Clearly, if $b>0$ and $0<a<<1$, the dominant term is $-b a^2 \f{\|\nabla \phi\|^2}{2}	- \f{a^{\f{d(p+1)}2-d}}{p+1}  \|\phi\|_{L^{p+1}}^{p+1}<0$, whence $m(\la)<0$ for these values.
	
	On the other hand,  if $p>1+\f{8}{d}$, we have $\f{d(p+1)}2-d>4$, so that
	$\lim_{a\to +\infty} I[\phi_a]=-\infty$.
\end{proof}
The next lemmata are  technical statements,  which will however impact the restrictions one must impose on $p$ (and other parameters),
in order to be able to construct the waves in Theorem \ref{theo:10}. In fact, we shall need specific Gagliardo-Nirenberg-Sobolev (GNS) type inequalities in order to resolve the existence requirements in Theorem \ref{theo:10} and Theorem \ref{theo:12}.
\subsection{The Gagliardo-Nirenberg-Sobolev  inequalities with mixed dispersion}
We start with the isotropic case.
\begin{proposition}
	\label{prop:20}
Let $b>0$. For every $d\geq 2$, 	there exists $p_*(d)$, so that: for all $1<p\leq p_*(d)$, the following estimate
\begin{equation}
\label{gns:10}
\|\phi\|_{L^{p+1}(\rd)}^{p+1}\leq C \|\phi\|_{L^2}^{p-1}
\left(\int_{\rd} [|\De \phi|^2 - b |\nabla\phi|^2 + \f{b^2}{4} \phi^2]dx\right)
\end{equation}
cannot hold for a given constant $C$ and all test functions $\phi$. In addition,
 $p_*(d)$ obeys the following
\begin{equation}
\label{g:10}
1+\f{4}{d}\leq  p_*(d)\leq 1+\f{8}{d+1}.
\end{equation}

On the other hand, for $1+\f{8}{d}>p>p_*(d)$, there exists a constant $C=C_{p,d,b}$, so that
\begin{equation}
\label{gns:20}
\|\phi\|_{L^{p+1}(\rd)}^{p+1}\leq C_{p,d,b} \|\phi\|_{L^2}^{p-1}
\left(\int_{\rd} [|\De \phi|^2 - b |\nabla\phi|^2 + \f{b^2}{4} \phi^2]dx\right)
\end{equation}

\end{proposition}
{\bf Remark:} The value of $p_*(1)=5$ was computed in \cite{atanas}. Finding the exact value of $p_*(d), d\geq 2$ appears to be a hard problem in Fourier analysis, closely related to the restriction conjecture. Even in our proof of the upper bound in  \eqref{g:10}, we use the full strength of the Stein-Tomas restriction theorem in two spatial dimensions (see for example p. 784, \cite{G1})
which does not appear to be enough to determine $p_*(d)$.
Proposition \ref{prop:20} allows us to prove Theorem \ref{theo:10}; see Section \ref{sec:3.4} below.

Next, we present the relevant GNS results (or lack thereof)  in the  anisotropic case.  The result is much more definite than its counterpart Proposition \ref{prop:20}.
\begin{proposition}
	\label{prop:25}
	Let $b>0$. For every $d\geq 2$, 	and  for all $1<p\leq 1+\f{4}{d}$, the following estimate
	\begin{equation}
	\label{gns:12}
	\|\phi\|_{L^{p+1}(\rd)}^{p+1}\leq C \|\phi\|_{L^2}^{p-1}
	\left(\int_{\rd} [|\De \phi|^2 - b |\p_{x_1} \phi|^2 + \f{b^2}{4} \phi^2]dx\right)
	\end{equation}
	cannot hold for a given constant $C$ and all test functions $\phi$.
	
	On the other hand, for $1+\f{8}{d}>p>1+\f{4}{d}$, there exists a constant $C=C_{p,d,b}$, so that
	\begin{equation}
	\label{gns:22}
	\|\phi\|_{L^{p+1}(\rd)}^{p+1}\leq C_{p,d} \|\phi\|_{L^2}^{p-1}
	\left(\int_{\rd} [|\De \phi|^2 - b |\p_{x_1} \phi|^2 + \f{b^2}{4} \phi^2]dx\right)
	\end{equation}
\end{proposition}

 In the remainder of this section, we present some  preparatory material for the proofs of Proposition \ref{prop:20} and Proposition \ref{prop:25}.
To this end, we use the formula for the Fourier transform and its inverse as follows
$$
\hat{f}(\xi)=  (2\pi)^{-d/2}  \int_{\rd} f(x) e^{-   i x\cdot \xi} dx, \ \ f(x)=  (2\pi)^{-d/2}
\int_{\rd} \hat{f}(\xi)  e^{  i x\cdot \xi} d\xi
$$
The Plancherel's theorem states that $\|f\|_{L^2}= \|\hat{f}\|_{L^2}$. We will also make frequent use of the Bernstein inequality: for every $1\leq p\leq q\leq \infty$ and every finite volume set $A\subset \rd$, there exists $C=C_d$, so that
$$
\|P_A f\|_{L^q}\leq C |A|^{\f{1}{p}-\f{1}{q}} \|f\|_{L^p}.
$$
where $\widehat{P_A f}(\xi) =\chi_A (\xi) \hat{f}(\xi)$.

We are now ready to proceed to the specifics of the isotropic case.
\subsection{Proof of Proposition \ref{prop:20}}
In consideration of the estimates \eqref{gns:20}, one can straightforwardly
rescale to the case $b=2$, which we will henceforth
use in our considerations for simplicity
(although when completing the proof of our theorems in section
IV below, we will present them for arbitrary $b$).
Using Fourier transformation and  Plancherel's theorem, we can rewrite
\begin{eqnarray*}
\int_{\rd} [|\De \phi|^2 - 2 |\nabla\phi|^2 + \phi^2]dx=
\int_{\rd} |\hat{f}(\xi)|^2 ( |\xi|^2- 1)^2 d\xi.
\end{eqnarray*}
Further, one can use smooth  decompositions near $|\xi|=1$  to study \eqref{gns:20}. More concretely,  introduce  a function $\psi\in C^\infty_0(\rone)$, so that $\psi(z)=1, |z|<1$ and $\psi(z)=0, |z|>2$. Then,
let  $\chi(z)=\psi(z)-\psi(2z)$,
so that $supp \chi\subset (\f{1}{2}, 2)$ and
$\sum_{j=-\infty}^\infty \chi(2^j z)=1, z\neq 0$. Now, introduce two multipliers
$$
\widehat{Q_j f}(\xi):=\chi(2^{-j}(|\xi|^2-1)) \hat{f}(\xi),
\widehat{P_m f}(\xi):=\chi(2^m(|\xi|^2-1)) \hat{f}(\xi),
$$
and the corresponding versions $Q_{>j}:=\sum_{l>j} Q_l$, $Q_{\sim j} =Q_{j-1}+Q_j+Q_{j+1}$ and so on.
Based on the relevant decomposition,
$$
Id=\sum_{j=0}^\infty Q_j + \sum_{m>0} P_m,
$$
and  $Q_{j}, j\geq 3$ Fourier restricts to a region $|\xi|\sim 2^{j/2}$. We henceforth adopt the notation, $A\sim B$, for two positive quantities that satisfy  if $\f{1}{4} A\leq B \leq 4 A$.

 It is actually not hard to come up with necessary and sufficient conditions on $p$
so that \eqref{gns:20} holds, where $f$ is replaced by $Q_{>3} f$.
\subsubsection{Estimates away from $|\xi|=1$}
We can estimate by Sobolev embedding (or rather Bernstein inequality)
\begin{eqnarray}
\label{il:10}
\|Q_j f\|_{L^{p+1}(\rd)}^{p+1} \leq C 2^{j\f{(p-1)d}{4}}  \|Q_j f\|_{L^2(\rd)}^{p+1}.
\end{eqnarray}
Computing the right-hand side of \eqref{gns:20} (with $b=2$),  on the other hand, yields
$
2^{2j} \|Q_j f\|_{L^2(\rd)}^{p+1}.
$
One can now show \eqref{gns:20} for $Q_{>3} f$, when $1<p<1+\f{8}{d}$.  Indeed, by the triangle inequality,  \eqref{il:10}
\begin{eqnarray*}
\|Q_{>3} f\|_{L^{p+1}(\rd)} &\leq & C \sum_{j>3} 2^{j\f{(p-1)d}{4(p+1)}}
 \|Q_j f\|_{L^2(\rd)}\leq C \|f\|^{\f{p-1}{p+1}} \left(\sum_{j>3} \|Q_j f\|_{L^2}^2
2^{j\f{(p-1)d}{4}} \right)^{\f{1}{p+1}}\leq C \|f\|^{\f{p-1}{p+1}} \left(\sum_{j>3} \|Q_j f\|_{L^2}^2
2^{2j} \right)^{\f{1}{p+1}}\\
&\leq & C \|f\|^{\f{p-1}{p+1}} \left(\int_{\rd} [|\De f|^2 - 2 |\nabla f|^2 + f^2]dx\right)^{\f{1}{p+1}}.
\end{eqnarray*}
where we have used $\f{(p-1)d}{4}<2$ and $||\xi|^2-1|\sim 2^{j}$ on the support of the multiplier $Q_j$.

The situation is much more delicate for frequencies close to the sphere $|\xi|=1$, that is for the multipliers $P_m, m>>1$.
\subsubsection{Estimates near $|\xi|=1$}
Clearly, one has, by Bernstein inequalities, the estimates $\|P_m f\|_{L^{p+1}}\leq C_m \|f\|_{L^2}$, so the issue is the control of $P_{>m}$ for a fixed $m$.
We claim that the central issue here is the exact bound in the estimate
$\|P_m f\|_{L^{p+1}(\rd)}\leq C \|f\|_{L^2}$. More precisely, define
\begin{equation}
\label{ex:10}
\al(p,d)=\sup\{\al:  \limsup\limits_{m\to \infty} \sup_{\|f\|_{L^2}=1}
2^{\al m} \|P_m f\|_{L^{p+1}(\rd)}<\infty \}.
\end{equation}
Note that by the uniform boundedness principle and the definition of $\al(p,d)$,  for every $\be>\al(p,d)$, there is a $f_\be: \|f_\be\|_{L^2}=1$, so that
\begin{equation}
\label{ex:20}
\limsup\limits_{m\to \infty}
2^{\be m} \|P_m f\|_{L^{p+1}(\rd)}=\infty.
\end{equation}
For convenience, we drop the dependence on the dimension $d$ in $\al(p,d)$. Note that by the Bernstein's inequality $\al(p)>0$, in fact $\|P_m f\|_{L^{p+1}}\leq C 2^{-m(\f{1}{2}-\f{1}{p+1})} \|f\|_{L^2}$, whence
$\al(p,d)\geq (\f{1}{2}-\f{1}{p+1})$. For the same reasons, it is clear that $p\to \al(p)$ is an increasing function. In addition $p\to \al(p)$ is a continuous function and $\al(1)=0$.

A convenient characterization of $\al(p,d)$ is the following:   for every $\eps>0$, there is $C_\eps$, so that
\begin{equation}
\label{ex:35}
\|P_m f\|_{L^{p+1}(\rd)}\leq C_\eps 2^{(-\al(p,d)+\eps)m} \|f\|_{L^2(\rd)}.
\end{equation}
We claim that  $\al(p,d)$ determines the value of $p_*(d)$ in Proposition \ref{prop:20}.

 In fact, we claim that $p_*(d)$ is the unique solution of the equation $\al(p,d)=\f{2}{p+1}$. We now prove this claim. First, we show that this equation has a unique solution. To start with, the continuous function $h(p):=\al(p)-\f{2}{p+1}$ is increasing, with $h(1)=-1<0$, while
 $h(p)\geq \f{1}{2}-\f{3}{p+1}>0$ for $p\geq 5$, so there will be a solution\ $p\in (1,5)$. In fact, below we show better bounds on $\al(p,d)$, which
  imply existence of solutions in the interval of interest, namely $(1,1+\f{8}{d})$, but the existence of solutions anywhere in $(1, \infty)$ will suffice for now.

 Next, we show that for $p<p_*(d)$, \eqref{gns:10} holds.  This means that $\al(p)-\f{2}{p+1}<0$.  Introduce $\be>\al(p)$, so that $\be<\f{2}{p+1}$. According to the remarks made earlier,
 this allows us to find a function $f_\be:\|f_\be\|_{L^2}=1$, so that \eqref{ex:20} holds true. Assume then, for a contradiction, that \eqref{gns:10} holds for some constant $C$. This means that for all $\phi\neq 0$,
 \begin{equation}
 \label{ex:30}
  \f{\|\phi\|_{L^{p+1}}}
 {\|\phi\|_{L^2}^{\f{p-1}{p+1}} \left(\int [|\De \phi|^2 - 2 |\nabla\phi|^2 +   \phi^2]dx\right)^{\f{1}{p+1}}}\leq C.
 \end{equation}
 In particular, taking into account the properties of $P_m$ and our earlier calculations with regards to the quantity in the denominator of \eqref{ex:30}, we can take $\phi_m=P_m f_\be$,
  \begin{equation}
  \label{ex:40}
  \sup_m  \f{\|P_m f_\be\|_{L^{p+1}}}
  {\|P_m f_\be\|_{L^2} 2^{-\f{2 m}{p+1}}}\leq C.
  \end{equation}
  Now, since $\|P_m f_\be\|_{L^2}\leq \|f_\be\|_{L^2}=1$, it follows that
  $\sup_m  \|P_m f_\be\|_{L^{p+1}} 2^{\f{2 m}{p+1}}\leq C$. But this is a contradiction with \eqref{ex:20}, since
  $$
  \|P_m f_\be\|_{L^{p+1}} 2^{\be m} 2^{m(\f{2}{p+1}-\be)} =  \|P_m f_\be\|_{L^{p+1}} 2^{\f{2 m}{p+1}}\leq C,
  $$
whereas on the left hand side $\limsup_m  \|P_m f_\be\|_{L^{p+1}} 2^{\be m} =\infty$, and $\lim_m 2^{m(\f{2}{p+1}-\be)} =\infty$.

Assume now $p>p_*(d)$, so $\al(p)>\f{2}{p+1}$. So, we can find $\be:
\al(p)>\be>\f{2}{p+1}$. Then, we have the estimate, (see \eqref{ex:35}, but applied to $P_m^2 f$)
\begin{equation}
\label{ex:37}
\|P_m f\|_{L^{p+1}(\rd)}\leq C_\eps 2^{-\be m} \|P_m f\|_{L^2(\rd)}.
\end{equation}
We will show that \eqref{gns:20} holds. In view of the estimates away from $|\xi|=1$, which establish \eqref{gns:20} for $Q_{>3} f$, it suffices to consider $P_{>m_0} f$ for $m_0$ sufficiently large only. We take $m_0=10$ for concreteness. By \eqref{ex:37}, we have
\begin{eqnarray*}
\|P_{>10} f\|_{L^{p+1}} &\leq & \sum_{m>10} \|P_{m} f\|_{L^{p+1}}\leq
C \sum_{m>10} 2^{-\be m} \|P_{m} f\|_{L^2}  \leq  C \|f\|_{L^2}^{\f{p-1}{p+1}} \left(\sum_{m>10} \|P_{m} f\|_{L^2}^2 2^{-2m} \right)^{\f{1}{p+1}} \leq \\
&\leq & C \|f\|_{L^2}^{\f{p-1}{p+1}}  \left(\int [|\De f|^2 - 2 |\nabla f|^2 +   f^2]dx\right)^{\f{1}{p+1}}.
\end{eqnarray*}
This establishes \eqref{gns:20} and so the estimates (or lack thereof) in Proposition \ref{prop:20} are  established in full.

We now focus our attention on appropriate estimates on $p_*(d)$.
\subsubsection{Estimates on the value of $p_*(d)$}
As we saw above, the value $p_*(d)$ is intimately related to the precise estimates of $P_m:L^2(\rd) \to L^{p+1}(\rd)$. Recall that there was the trivial bound based on the Bernstein's inequality, $\al(p,d)\geq (\f{1}{2}-\f{1}{p+1})$,  but we now aim for a much more sophisticated one.
Before we proceed,
we need to introduce some quantities that will be helpful in our considerations. The surface measure on $\sd$  is defined via $d \si(x)=\de(|x|^2-1)$ and its Fourier transform is  (see \cite{G1}, Appendix B.4)
$$
\widehat{d\si}(\xi)=(2\pi)^{-d/2} \int_{\sd} e^{- i \xi\cdot  \theta} d\theta=c_d \f{J_{\f{d-2}{2}}(|\xi|)}{|\xi|^{\f{d-2}{2}}}=: S(\xi).
$$
where $c_d$ is a constant and $J_n$ are the standard Bessel functions. Furthermore, see Appendix B.5, \cite{G1}, for any radial function $f(x)=f_0(|x|)$, one can compute its Fourier transform as follows
$$
\hat{f}(\xi)= C_d |\xi|^{-\f{d-2}{2}} \int_0^\infty f_0(r) J_{\f{d-2}{2}}(r|\xi|) r^{\f{d}{2}} dr.
$$
In this way, when we take the multipliers associated to $P_m$, namely $f_0(r)=\chi(2^m (r^2-1))$, we see that its kernel $K_m$ (i.e., $P_m f=f*K_m$)
can be expressed in terms of an averaging operator involving the kernel $S=\widehat{d\si}$ as follows
$$
K_m(x)= C_d \int_0^\infty \chi(2^m (r^2-1)) \f{J_{\f{d-2}{2}}(r|x|)}{|x|^{\f{d-2}{2}}} r^{\f{d}{2}} dr=C_d \int_0^\infty
\chi(2^m (r^2-1))  r^{d-1} S(r|x|) dr.
$$
Let us now fix $q>2$. We wish to establish an estimate for the operator norm $P_m: L^2\to L^q$.
Note that $P_m$ is trivially bounded, but the  issue is to  determine precise bounds on the norm, as a function of $m$.   Due to the fact that $\chi$ is real-valued, $P_m: L^{q'}\to L^2$, so in order to compute $\|P_m\|_{B(L^2\to L^q)}$, we might instead consider $P_m^2:L^{q'}\to L^q$ and in addition
$$
\|P_m\|_{B(L^2\to L^q)}=\sqrt{\|P_m^2\|_{B(L^{q'}\to L^q)}}.
$$
Now, $P_m^2 f=\tilde{K}_m*f$, $\tilde{K}_m=C_d \int_0^\infty
\chi^2(2^m (r^2-1))  r^{d-1} S(r|x|) dr$. The advantage in this formulation is that the mapping properties of the operator $f\to f*\widehat{d\si}=f*S$ are well-understood. In fact,  this is the content of the celebrated Stein-Tomas theorem. To summarize, (see $(10.4.7)$, p. 784, \cite{G1}), there is the  estimate
\begin{equation}
\label{st}
\|f*S\|_{L^{q}(\rd)}\leq C \|f\|_{L^{q'}(\rd)}, q_d=\f{2(d+1)}{d-1}
\end{equation}
With this value of $q_d$, we then conclude that since
$
f*\tilde{K}_m=C_d \int_0^\infty
\chi^2(2^m (r^2-1))  r^{d-1} [S(r|\cdot|) *f] dr,
$
we have the estimate, based on the Stein-Tomas bound \eqref{st},
$$
\|P_m^2 f\|_{L^q}=\|f*\tilde{K}_m\|_{L^q}\leq C 2^{-m} \|f\|_{L^{q'}}.
$$
Note here that the factor $2^{-m}$ is gained through the integration in $r$, while the estimate for the term
$\|[S(r|\cdot|) *f]\|_{L^q}$ comes from \eqref{st}.
Accordingly, this gives the estimate
\begin{equation}
\label{o:19}
\|P_m f\|_{L^{q_d}}\leq C 2^{-m/2} \|f\|_{L^2}, q_d=\f{2(d+1)}{d-1}.
\end{equation}
Interpolating this estimate with the trivial one $\|P_m f\|_{L^2}\leq C \|f\|_{L^2}$, we conclude
that for every $2\leq q\leq \f{2(d+1)}{d-1}$, there is
\begin{equation}
\label{o:20}
\|P_m f\|_{L^q}\leq C 2^{-m \f{d+1}{2}\left(\f{1}{2}-\f{1}{q}\right)} \|f\|_{L^2}.
\end{equation}
It follows that
\begin{equation}
\label{o:25}
\al(p,d)\geq \f{d+1}{2}\left(\f{1}{2}-\f{1}{p+1}\right).
\end{equation}
In particular, it is clear that $\al(p,d)-\f{2}{p+1}>0$, if $p>1+\f{8}{d+1}$, which means that we have established the upper bound $p_*(d)< 1+\f{8}{d+1}$.

In order to establish the lower bound for $p_*(d)$, we test the ratio
$
\f{\|P_m f\|_{L^r(\rd)}}{\|f\|_{L^2(\rd)}}
$
for $r>2$, with  $f=K_m$, defined above. For the correct asymptotics, we need to recall (see Appendix B.6, \cite{G1})  that for every $r>>1$,
$$
J_k(r)=c \f{\cos(r-\f{\pi k}{2}-\f{\pi}{4})}{\sqrt{r}}+O(r^{-3/2}).
$$
Now,
\begin{eqnarray*}
\tilde{K}_m(x) &=&C_d \int_0^\infty
\chi^2(2^m (r^2-1))  r^{d-1} S(r|x|) dr=  \\
&=& const. |x|^{-\f{d-1}{2}} \int_0^\infty
\chi^2(2^m (r^2-1))  r^{d-1}  \f{\cos(r|x|-\f{\pi (d-2)}{4}-\f{\pi}{4})}{r^{\f{d-1}{2}}}dr + 2^{-m} O(|x|^{-\f{d+1}{2}}).
\end{eqnarray*}
It is then easy to see  that for $2^{-m}<<\de<<1$ and $|x|\sim \de 2^{m}, m>>1$, in the integral above there is the approximate formula
$$
\cos\left(r|x|-\f{\pi (d-2)}{4}-\f{\pi}{4}\right)= \cos\left(|x|-\f{\pi (d-2)}{4}-\f{\pi}{4}\right)+ O(\de).
$$
This implies that for a fixed portion of the set $|x|\sim \de 2^{m}$, $\cos(r|x|-\f{\pi (d-2)}{4}-\f{\pi}{4})\geq \f{1}{2}$, whence we have that
$\tilde{K}_m$ obeys, on this fixed portion of the set, the bound
$
|\tilde{K}_m(x)|\gtrsim 2^{-m\f{d+1}{2}}.
$
Thus,
$$
\|P_m f\|_{L^r(\rd)}\geq c 2^{-\f{m}{2}} 2^{-md\left(\f{1}{2}-\f{1}{r}\right)}
$$
while by Plancherel's theorem
$$
\|f\|_{L^2}=\|K_m\|_{L^2}=(\int_{\rd} |\chi(2^m(|\xi|^2-1))|^2 d\xi)^{\f{1}{2}}\sim 2^{-\f{m}{2}}.
$$
Thus,
$$
\f{\|P_m f\|_{L^r(\rd)}}{\|f\|_{L^2(\rd)}} \geq c 2^{-md\left(\f{1}{2}-\f{1}{r}\right)}.
$$
It follows that one has the inequality complementary to \eqref{o:25},
\begin{equation}
\label{o:27}
\al(p,d)\leq d\left(\f{1}{2}-\f{1}{p+1}\right).
\end{equation}
We can now derive an estimate for $p_*(d)$. Indeed,
$$
\al(p,d)-\f{2}{p+1}\leq d\left(\f{1}{2}-\f{1}{p+1}\right)-\f{2}{p+1}<0,
$$
if $p<1+\f{4}{d}$. Thus, we conclude that $p_*(d)>1+\f{4}{d}$. This finishes the proof of Proposition \ref{prop:20}.

Our next goal is to analyze  the relevant GNS inequalities in the non-isotropic case.

\subsection{The anisotropic case: Proof of    Proposition \ref{prop:25}}
Again, a simple rescaling argument reduces matters to the case $b=2$, as in the proof of Proposition \ref{prop:20}.
The arguments for the anisotropic case are pretty similar, once we realize the important differences in the dispersion relations. More specifically, using Plancherel's theorem
in this case:
\begin{eqnarray*}
\int_{\rd} [|\De \phi|^2 - 2 |\p_{x_1} \phi|^2 +  |\phi|^2]dx =
\int_{\rd} |\hat{\phi}(\xi)|^2[(\xi_1^2-1)^2+|\xi'|^4+2\xi_1^2 |\xi'|^2] d\xi,
\end{eqnarray*}
where $\xi'=(\xi_2, \ldots, \xi_d)$. For future reference, introduce the dispersion related function
$h(\xi):= (\xi_1^2-1)^2+|\xi'|^4+2\xi_1^2 |\xi'|^2$. Based on this formula, we discuss the validity of \eqref{gns:12}.

We start our analysis by considering some easy regions. One such region, is when $\xi_1$ is away from $\pm 1$. Quantitatively, $|\xi_1^2-1|\geq \f{1}{100}$ and say $f_0:=P_{|\xi_1^2-1|\geq \f{1}{100}} f$. In this case, we clearly have $h(\xi)\sim 1+ \xi_1^4+|\xi'|^4\sim <\xi>^4$.
In such a scenario, it is easy to analyze $\|f_0\|_{L^{p+1}(\rd)}$, in particular what it takes for \eqref{gns:22} to hold (and \eqref{gns:12} to fail respectively).

More specifically, assuming $1<p<1+\f{8}{d}$, we have
by Bernstein's inequality and Plancherel's
equality
\begin{eqnarray*}
	\|f_0\|_{L^{p+1}(\rd)}^{p+1}  &\leq &  C(\|f_0\|_{L^2}+ \sum_{k=0}^\infty 2^{k d(\f{1}{2}-\f{1}{p+1})}
	\|P_k f_0\|_{L^2})^{p+1} \leq
	C \|f_0\|_{L^2}^{p-1} \sum_{k=0}^\infty 2^{4k} \int_{|\xi|\sim 2^k } |\hat{f_0}(\xi)|^2 d\xi\\
	&\leq & C\|f\|_{L^2}^{p-1} \int_{\rd} |\hat{f_0}(\xi)|^2 h(\xi) d\xi \leq
	C\|f\|_{L^2}^{p-1} \int_{\rd} [|\De f|^2 - 2 |\p_{x_1} f|^2 +  |f|^2]dx.
\end{eqnarray*}
This shows that $1<p<1+\f{8}{d}$ is a sufficient condition for the validity of \eqref{gns:22}, in the case, where $\xi_1$ is away from $\pm 1$. On the other hand, testing \eqref{gns:22} with a function of the type $\hat{f}(\xi)=\vp(2^{-k} \xi)$ for $k>>1$, shows that $1<p<1+\f{8}{d}$ is necessary as well.

We now turn our attention to the more interesting cases, namely $|\xi_1^2-1|\sim 2^{-m}$, $m>>1$. In this case,
$$
h(\xi)\sim 2^{-2m} + |\xi'|^4+|\xi'|^2.
$$
The case $|\xi'|\geq \f{1}{100}$   reduces to $h(\xi)\sim <\xi>^4$, which was just analyzed. So, it remains to consider the cases
$|\xi'|< \f{1}{100}$. So, the dispersion relation will be exactly
$h(\xi)\sim 2^{-2m}+|\xi'|^2$. Further, by changing the Fourier
variables $\xi_1\to \xi_1\pm 1$ (which on the physical side means replacing $f$ with
$f\to e^{\mp i x_1} f$, a harmless operation in terms of all $\|\cdot\|_{L^q}$ norms), we are reduced to studying the question: for which values of $p$ can  the inequality hold
\begin{equation}
\label{jo}
\|f\|_{L^{p+1}}^{p+1} \leq C \|f\|_{L^2}^{p-1} \int_{\rd} |\hat{f}(\xi)|^2 |\xi|^2 d\xi,
\end{equation}
where $f$ is a function $supp \hat{f}\subset \{\xi: |\xi|<<1\}$.

We will now show that $p\geq 1+\f{4}{d}$ is a necessary and sufficient condition for \eqref{jo} to hold. We have already established that $1+\f{8}{d}>p$ is necessary and sufficient for the region away from $\xi_1=\pm 1$, which will of course need to be intersected with the necessary and sufficient condition for \eqref{jo} to hold.

To this end, assume that
$p\geq 1+\f{4}{d}$. Consider $f=\sum_{k=0}^\infty P_{-k}  f $ (recall $supp \hat{f}\subset \{\xi: |\xi|<<1\}$), so by standard properties of the Littlewood-Paley decompositions and the Bernstein's inequality
$$
\|f\|_{L^{p+1}(\rd)}^2\leq C \sum_{k=0}^\infty \|P_{-k}  f\|_{L^{p+1}(\rd)}^2\leq C
\sum_{k=0}^\infty 2^{-2kd(\f{1}{2}-\f{1}{p+1})} \|P_{-k}  f\|_{L^2(\rd)}^2.
$$
Further applying Cauchy-Schwartz
\begin{eqnarray*}
\|f\|_{L^{p+1}(\rd)}^2 &\leq &  C \left(\sum_{k=0}^\infty \|P_{-k} f\|_{L^2}^2 \right)^{\f{p-1}{p+1}}
\left(\sum_k 2^{-2kd(\f{1}{2}-\f{1}{p+1})\f{p+1}{2}} \|P_{-k}  f\|_{L^2(\rd)}^2\right)^{\f{2}{p+1}}\\
&\leq & C \|f\|_{L^2}^{\f{p-1}{p+1}} \left(\sum_k 2^{-kd\f{p-1}{2}} \|P_{-k}  f\|_{L^2(\rd)}^2\right)^{\f{2}{p+1}}.
\end{eqnarray*}
It follows that whenever $p\geq 1+\f{4}{d}$, we have
$$
 \|f\|_{L^{p+1}}^{p+1}\leq C \|f\|_{L^2}^{p-1}
\sum_{k=0}^\infty 2^{-2k} \|P_{-k}  f\|_{L^2(\rd)}^2
\leq  C\|f\|_{L^2}^{p-1} \int_{\rd} |\hat{f}(\xi)|^2 |\xi|^2 d\xi.
$$
This establishes \eqref{jo} under the assumption $p\geq 1+\f{4}{d}$.
Conversely, assuming that \eqref{gns:20} holds, we test it with a function $f: \hat{f}(\xi)= \chi(2^k(\xi_1-1), 2^k \xi'), k>>1$. This yields the inequality $p\geq 1+\f{4}{d}$. Thus, we have finally established that the necessary and sufficient condition for \eqref{gns:22} to hold is exactly $1+\f{8}{d}>p\geq 1+\f{4}{d}$.

 \section{Completion of the proofs of Theorem \ref{theo:10} and Theorem \ref{theo:12}}
 \label{sec:3.4}
 We start our presentation with the proof for the existence of the waves. Along the way, we establish  a few necessary spectral properties of the corresponding linearized operators, which will be instrumental in the spectral stability considerations.
 \subsection{Existence of the waves - isotropic case}
 \label{sec:4.1}
 In this section, we present the proofs for the existence (or at least a very detailed scheme of the proof) for the isotropic case.
 We start with a few words about strategy, even though our approach, in principle, is
 a quite natural one.
 It was established in Lemma \ref{le:10} that the constrained minimization problem \eqref{e:10} is well-posed, and in fact $-\infty<m(\la)<0$. We would like to show that there is
 a minimizer for this problem, which subsequently will be shown to satisfy the Euler-Lagrange equation \eqref{47}. To this end,
consider  a minimizing sequence, say $\phi_k\in H^2(\rd)$. Ultimately, we would like to show that a strongly convergent subsequence of $\phi_k$ will converge to a solution $\Phi$.
The central  issue that we need to address is the non-triviality of such a minimizing sequence. This is the subject of the next  technical lemma.
\begin{lemma}
	\label{prop:29}
	Let $b>0, d\geq 2$ and $1<p<1+\f{8}{d}$. Let also
	\begin{itemize}
		\item $1<p\leq p^*(d)$ and $\la>0$
		\item $p^*(d)<p<1+\f{8}{d}$ and $\la>\la_{b,p,d}$.
	\end{itemize}
	Then, there exists a subsequence of $\phi_k$ so that for some $L_1>0, L_2>0, L_3>0$,
	\begin{equation}
	\label{120}
	\int_{\rd} |\De \phi_k|^2 dx \to L_1; 	\int_{\rd} |\nabla \phi_k|^2 dx \to L_2;  \ \ 	
	\int_{\rd} |\phi_k|^{p+1} dx \to L_3.
	\end{equation}
\end{lemma}
Informally, the claim is that for $1<p\leq p^*(d), \la>0$ and for $p^*(d)<p<1+\f{8}{d}, \la>\la_{b,p,d}$ (where $\la_{b,p,d}$ is some threshold depending on the parameters $b,p,d$), one has non-trivial minimizing sequences. Note that this does not yet show the existence of a limit, for which we bring the full weight of the compensated compactness theory to bear. At the same time this rules out some of the main obstacles toward the strong convergence of,  a subsequence of a translate of $\phi_k$,  to a minimizer.
\begin{proof}
	By the estimate \eqref{112}, it is clear that $\{	\int_{\rd} |\De \phi_k|^2 dx\}_k$ is a bounded sequence. Since $\|\phi_k\|^2=\la$ is fixed, by Sobolev embedding it follows that
	$\int_{\rd} |\nabla \phi_k|^2 dx, \int_{\rd} |\phi_k|^{p+1} dx $ are bounded as well. We can take subsequences to ensure that  the convergences in \eqref{120} hold true.
	
	Now, it remains to establish  the non-trivial claim,
	namely that all $L_1, L_2, L_3$ are non-zero. Assume for a contradiction that $L_3=0$. Introduce
	$$
	\tilde{I}[u]:=\left\{\f{1}{2}  \int_{\rd} |\De u|^2- \f{b}{2} \int_{\rd} |\nabla u|^2 \right\}.
	$$
	Clearly, $\tilde{I}[u]\geq I[u]$, whereas
	$
\lim_k \tilde{I}[\phi_k]=\lim_k I[\phi_k]=	\inf_{\|u\|^2=\la} I[u]\leq 	\inf_{\|u\|^2=\la} \tilde{I}[u].
	$
	It follows that $\phi_k$ is a minimizing sequence for the problem $	\inf_{\|u\|^2=\la} \tilde{I}[u]$ and the minima coincide. On the other hand, by Plancherel's theorem
	\begin{equation}
	\label{125}
	\int_{\rd} |\De u|^2- b\int_{\rd} |\nabla u|^2 +\f{b^2}{4} \|u\|^2=\int_{\rd} |\hat{u}(\xi)|^2
	\left(|\xi|^2 - \f{b}{2}\right)^2 d\xi\geq 0,
	\end{equation}
	whence $\inf_{\|u\|^2=\la} \tilde{I}[u]\geq -\f{b^2}{8}\la$. In fact, there is equality, i.e.,
$	\inf_{\|u\|^2=\la} \tilde{I}[u]=-\f{b^2}{8}\la$ as the inequality in \eqref{125} may be saturated by choosing a function $u$, so that $\hat{u}$  is supported arbitrarily close to $|\xi|=\f{\sqrt{b}}{2}$.

All in all, it follows that  $\inf_{\|u\|^2=\la} I[u]= -\f{b^2}{8}\la$. Applying this to arbitrary $f\neq 0$, and then
$u=\sqrt{\la} \f{f}{\|f\|}$, so that $\|u\|^2=\la$, we have
$$
\f{2\la^{\f{p-1}{2}}}{p+1} \int_{\rd} |f|^{p+1}
\leq \|f\|^{p-1} \left\{ \int_{\rd} |\De f|^2- b  |\nabla f|^2+\f{b^2}{4} |f|^2 \right\}.
$$
This last inequality is in contradiction with \eqref{gns:10} for $1<p\leq p^*(d)$, and with \eqref{g:10} for all large enough $\la$. This completes the proof for $L_3>0$ under these assumptions.

Assuming that either $L_1=0$ or $L_2=0$ implies, by the standard Gagliardo-Nirenberg inequality,  the fact that $L_3=0$, which we have just shown to be impossible.
\end{proof}
The rest of the proof for existence of a minimizer proceeds identically to the one presented in Section 3.2, \cite{atanas}. Namely, first one establishes that the function $\la\to m(\la)$ is strictly sub-additive for $1<p<1+\f{8}{d}$. That is, for all $\al\in (0, \la)$,
\begin{equation}
\label{130}
m(\la)<m(\al)+m(\la-\al).
\end{equation}
This is standard, and proceeds via the  property that $\la\to \f{m(\la)}{\la}$ is a non-increasing function, which can be obtained via elementary scaling arguments and the crucial property
$\lim_k \int_{\rd} |\phi_k|^{p+1}=L_3>0$, which was established in \eqref{120}.

Next, taking  a minimizing subsequence $\phi_k$, with the property \eqref{120}, one applies the compensated  compactness lemma to it. More specifically,  by the P.L. Lions concentration compactness lemma  (see Lemma 1.1, \cite{lions}), applied to $\rho_k:=|\phi_k|^2\in L^1(\rd)$, $\|\rho_k\|_{L^1}=\la$ there is a subsequence (denoted again by $\rho_k$), so that  at least one of the following is satisfied:
\begin{enumerate}
	\item \textit{Tightness.} There exists $ y_k \in \rone $ such that for any $ \ve >0 $ there exists $ R(\ve) $ such that for all $ k $
	$$
	\int_{B(y_k,R(\ve))}\rho_k dx\geq \int_{\rone}\rho_k dx-\ve.
	$$
	\item \textit{Vanishing.} For every $ R>0 $
	$$
	\lim\limits_{k\to \infty } \sup_{y\in \rone}\int_{B(y,R)}\rho_kdx=0.
	$$
	\item \textit{Dichotomy.} There exists $ \alpha \in (0,\lambda) $, such that for any $ \ve>0 $ there exist $ R, R_k\to \infty,y_k\in\rd $,  such that
	\begin{equation}
	\label{350}
	\left\{
	\begin{array}{c}
	\left| \int\limits_{B(y_k,R)}\rho_k dx-\alpha \right| < \ve, \ \left| \int\limits_{R<|x-y_k|<R_k}\rho_k dx\right| < \ve, \\
	\left| \int\limits_{R_k<|x-y_k|}\rho_k dx  -(\la-\al)\right| < \ve.
	\end{array}
	\right.
	\end{equation}
\end{enumerate}
 Then, one shows that the dichotomy cannot occur. The proof proceeds via an argument that shows that dichotomy leads to a   inequality of the form $m(\la)\geq m(\al)+m(\la-\al)$, with $\al$ as in the dichotomy alternative. This of course contradicts the strict sub-additivity \eqref{130}.
 Next, vanishing leads, via the standard Gagliardo-Nirenberg's,  to $\lim_k \int_{\rd} |\phi_k|^{p+1}=0$,  in a contradiction with \eqref{120}, namely $\lim_k \int_{\rd} |\phi_k|^{p+1}=L_3>0$.

 Hence, one concludes tightness. But tightness means that for some sequence $y_k\in\rd$, there is strong $L^2$ convergence for $\{\phi_k(x-y_k)\}$. Denote $\Phi(\cdot) :=\lim_k \phi_k(\cdot-y_k)$. By the lower semi-continuity of the $L^2$ norm with respect to weak convergence,  we also conclude that $\lim_k \|\Phi(\cdot)- \phi_k(\cdot-y_k)\|_{H^2}=0$, whence $\Phi$ is a constrained minimizer of \eqref{e:10}, all under the assumptions of Lemma  \ref{prop:29}.

 Now, we take on the question for the Euler-Lagrange equation \eqref{47}. To this end,  fix a test function $h$ and  consider the scalar function
 $$
 g(\eps):=I\left(\sqrt{\la}\f{\Phi+\de h}{\|\Phi+\de h\|}\right).
 $$
 Since $g$ is differentiable in a neighborhood of the origin, and achieves its minimum there, we have that $g'(0)=0$. Since this is true for all test functions $h$, the resulting expression is that $\Phi$ is a distributional  solution of \eqref{47}. It is standard result in elliptic theory to conclude that such a solution $\Phi\in H^4(\rd)$. In addition, one can establish asymptotics for such functions, but we will not do so herein.

 Next, we consider the  second derivative necessary condition for a minimum at zero, which states that  $g''(0)\geq 0$. Assuming that the test function $h\perp \Phi$, we conclude $\dpr{\cl_+ h}{h}\geq 0$, which is exactly $\cl_+|_{\{\Phi\}^\perp}\geq 0$. In fact, this is sharp, because by a direct inspection \footnote{Note however that $\nabla \Phi\perp \Phi$} $\cl_+[\nabla \Phi]=0$. Also, since
 $
 \dpr{\cl_+\Phi}{\Phi}=-(p-1) \int |\Phi|^{p+1} dx<0,
 $
 we conclude that $\cl_+$ indeed has a negative eigenvalue. This coupled with $\cl_+|_{\{\Phi\}^\perp}\geq 0$ confirms that $\cl_+$ has exactly one negative eigenvalue.

 Finally, we show that $\cl_-\geq 0$. Assume not. Then, there is $\psi\perp \Phi, \|\psi\|=1,  \cl_-\psi=-\si^2 \psi$. Note however that $\cl_->\cl_+$, whence
 $$
 0\leq \dpr{\cl_+ \psi}{\psi}< \dpr{\cl_- \psi}{\psi}=-\si^2,
 $$
 which is a contradiction. Looking closely, this also shows that $0$ is a simple eigenvalue for $\cl_-$, because then, we take $\psi\perp \Phi: \cl_\psi=0$, and this still leads to a contradiction as above.
 Thus, we have shown the following proposition.
 \begin{proposition}
 	\label{prop:38}
 	Let $b>0, d\geq 2$, $1<p<1+\f{8}{d}$ and one of the two assumptions below are verified
 	\begin{itemize}
 		\item $1<p\leq p^*(d)$ and $\la>0$
 		\item $p^*(d)<p<1+\f{8}{d}$ and $\la>\la_{b,p,d}$.
 	\end{itemize}
 	Then, there exists $\Phi=\Phi_\la$, a constrained minimizer of \eqref{e:10} and  $\om=\om_\la>0$. In addition, $\Phi$ satisfies the Euler-Lagrange equation \eqref{47}, and the linearized Schr\"odinger operators $\cl_{\pm}$ satisfy
 	\begin{enumerate}
 		\item $\cl_-[\Phi]=0$, $0\in\si_{p.p.}(\cl_-)$ is a simple eigenvalue, and
 		$\cl_-|_{\{\Phi\}^\perp}\geq \de>0$, for some $\de>0$
 		\item $\cl_+|_{\{\Phi\}^\perp}\geq 0$. Moreover, $n(\cl_+)=1$
 	\end{enumerate}
 \end{proposition}
This completes the existence part of Theorem \ref{theo:10}.
\subsection{Existence of the waves - anisotropic case}
Following identical steps as in Section \ref{sec:4.1}, we establish the following analog of Lemma \ref{prop:29}.
 \begin{lemma}
 	\label{prop:39}
 	Let $b>0, d\geq 2$ and $1<p<1+\f{8}{d}$. Let also
 	\begin{itemize}
 		\item $1<p\leq 1+\f{4}{d}$ and $\la>0$
 		\item $1+\f{4}{d}<p<1+\f{8}{d}$ and $\la>\la_{b,p,d}$.
 	\end{itemize}
 	Then, there exists a subsequence of $\phi_k$ so that for some $L_1>0, L_2>0, L_3>0$,
 	\begin{equation}
 	\label{120}
 	\int_{\rd} |\De \phi_k|^2 dx \to L_1; 	\int_{\rd} |\nabla \phi_k|^2 dx \to L_2;  \ \ 	
 	\int_{\rd} |\phi_k|^{p+1} dx \to L_3.
 	\end{equation}
 \end{lemma}
 The proof of Lemma \ref{prop:39} proceeds in an identical manner  to the proof of Lemma \ref{prop:29} in Section \ref{sec:4.1}, with the suitable replacement of isotropic Proposition \ref{prop:20} with its anisotropic analog Proposition \ref{prop:25}.

 Once this step is completed, one establishes the strong sub-linearity of the cost function $n(\la)$, similar to the sub-linearity of $m(\la)$. The next step, again identical to the corresponding step for the isotropic case, is to show that once we take a minimizing sequence $\phi_k$, the method of compensated compactness goes through for the functions $\rho_k:=\phi_k^2$. This establishes the existence of the minimizer $\Phi$. Similarly, it satisfies the Euler-Lagrange equation and the spectral properties hold true. We collect the results in the next Proposition.
  \begin{proposition}
  	\label{prop:48}
  	Let $b>0, d\geq 2$, $1<p<1+\f{8}{d}$ and one of the two assumptions below are verified
  	\begin{itemize}
  		\item $1<p\leq 1+\f{4}{d}$ and $\la>0$
  		\item $1+\f{4}{d} < p <1+\f{8}{d}$ and $\la>\la_{b,p,d}$.
  	\end{itemize}
  	Then, there exists $\Phi=\Phi_\la, \om=\om_\la>0$, a constrained minimizer of \eqref{e:11}.
  	 In addition, $\Phi$ satisfies the Euler-Lagrange equation \eqref{22}
  	and the linearized   operators $\cl_{\pm}$ obey
  	\begin{enumerate}
  		\item $\cl_-\geq 0$. More specifically, $\cl_-[\Phi]=0$, $0\in\si_{p.p.}(\cl_-)$ is a simple eigenvalue, and
  		$\cl_-|_{\{\Phi\}^\perp}\geq \de>0$, for some $\de>0$
  		\item $\cl_+|_{\{\Phi\}^\perp}\geq 0$. Moreover, $n(\cl_+)=1$.
  	\end{enumerate}
  \end{proposition}

 \subsection{Spectral stability of the normalized waves}
 In this section,   we show the spectral stability of the waves constructed as
 constrained minimizers as \eqref{e:10}, \eqref{e:11} respectively.
 Starting with the eigenvalue problem \eqref{219}, we have that instability is equivalent to the solvability of the system
 \begin{equation}
 	\label{410}
 	 \left\{
 	\begin{array}{l}
 		\cl_- g=-\la f \\
 		\cl_+ f=\la g
 	\end{array}
 	\right.
 \end{equation}
 for some $\la: \Re\la>0$.
 So, applying $\cl_-$ to the second equation, we see that \eqref{410} reduces to the solvability of
  \begin{equation}
 	\label{411}
 	\cl_-\cl_+ f = -\la^2 f.
 \end{equation}
Conversely, if \eqref{411} has a non-trivial solution $\la, f$, then $g:=\la^{-1} \cl_+ f$ has a nontrivial solution $\la, f,g$.
So, \eqref{410} and \eqref{411}  are equivalent and we concentrate on the eigenvalue problem $\cl_-\cl_+ f = -\la^2 f$ henceforth.

 It follows immediately that $f\perp \Phi$. Thus, as $\cl_-|_{\{\Phi\}^\perp}\geq \de>0$, it follows that there exists unique $\eta\in \{\Phi\}^\perp$, so that $f=\sqrt{\cl_-} \eta$. Writing the relation $\cl_-\cl_+ f = -\la^2 f$ in terms of $\eta$ yields
 $$
 \sqrt{\cl_-}(\sqrt{\cl_-}\cl_+ \sqrt{\cl_-}\eta+\la^2 \eta)=0.
 $$
 As $\sqrt{\cl_-}\cl_+ \sqrt{\cl_-}\eta+\la^2 \eta\in \{\Phi\}^\perp=Ker(\cl_-)^\perp$, we conclude that
 $\sqrt{\cl_-}\cl_+ \sqrt{\cl_-}\eta+\la^2 \eta=0$. Thus,
  \begin{equation}
 	\label{412}
 \sqrt{\cl_-}\cl_+ \sqrt{\cl_-}\eta=-\la^2 \eta
\end{equation}
 Note however that the operator $\sqrt{\cl_-}\cl_+ \sqrt{\cl_-}$ is symmetric now, whence $-\la^2 \in \si_{p.p.}(\sqrt{\cl_-}\cl_+ \sqrt{\cl_-})$, so $-\la^2\in \rone$. We have already shown that there could not be oscillatory  instabilities. Furthermore, testing \eqref{412} with $\eta\in  \{\Phi\}^\perp$, we obtain
 $$
 -\la^2 \|\eta\|^2=\dpr{\cl_+ \sqrt{\cl_-} \eta}{\sqrt{\cl_-} \eta}=\dpr{\cl_+ f}{f}.
 $$
 Since $f\in  \{\Phi\}^\perp$ and $\cl_+|_{\{\Phi\}^\perp}\geq 0$, it follows that $\dpr{\cl_+ f}{f}\geq 0$, whence $-\la^2\geq 0$. This implies that all spectrum is stable, hence the spectral stability of $\Phi$ follows.

\section{Numerical Computations}

In the present section, we show a number of numerical computations for $d=2$ which corroborate and complement
our analytical results on the existence and stability of solitons for both the isotropic and anisotropic models with competing Laplacian and biharmonic operators.

\begin{figure}[htb]
    \centering
    \begin{tabular}{cc}
     \includegraphics[width=80mm,clip = true]{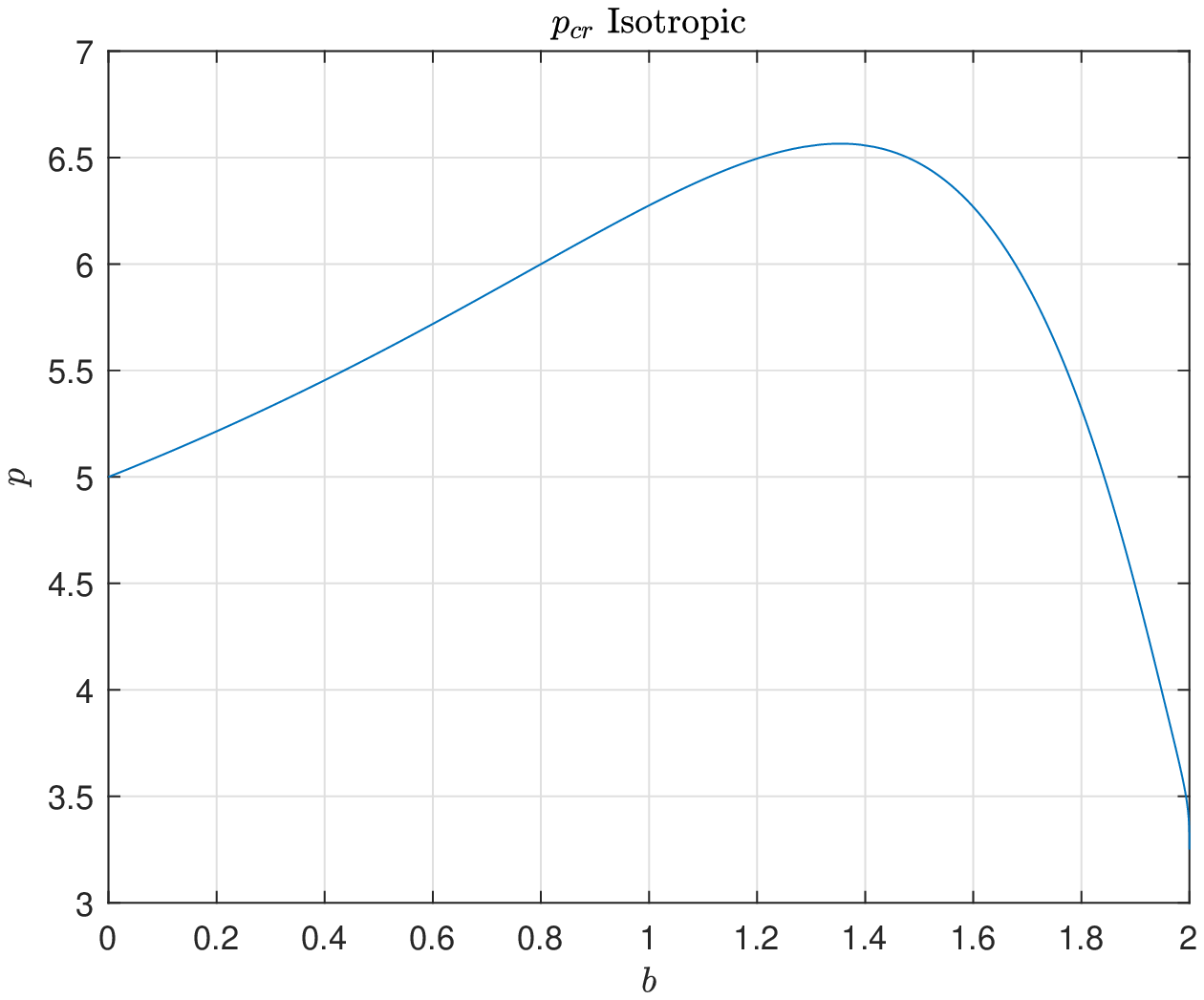} &
    \includegraphics[width=80mm,clip = true]{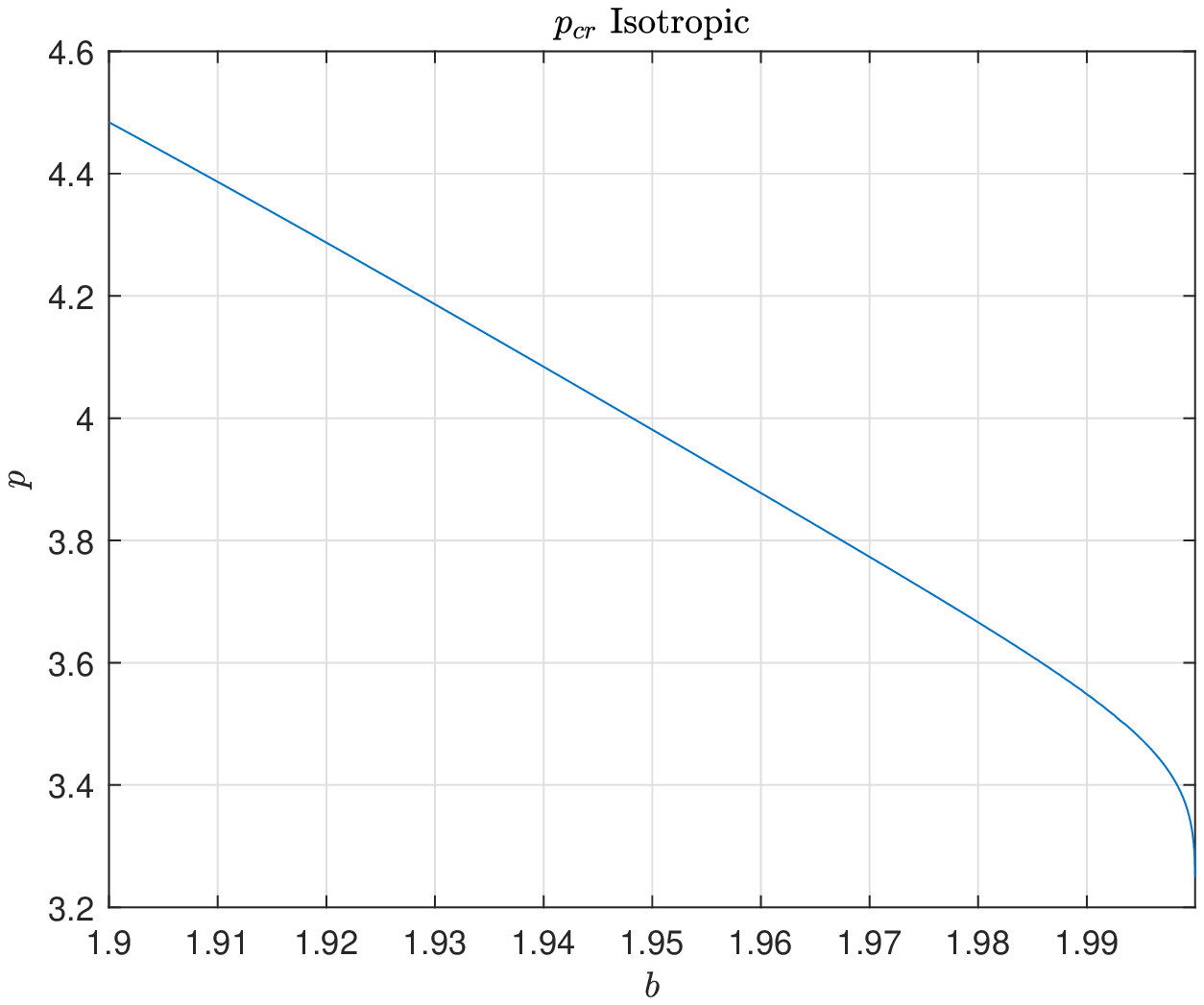}
    \end{tabular}
    \caption{Two-parameter plane of { the nonlinearity
    exponent parameter $p$ vs. the Laplacian prefactor $b$
    (varying between $0$ and $2$); recall that the frequency
    $\omega$ is fixed to unity, while our computations are for
    dimension $d=2$. The figure shows} the bifurcation loci separating spectrally stable solitons (under the curve) from unstable ones (above the curve). The right panel shows a blowup of the left one close to the edge point of $p=3$ and $b=2$.}
    \label{newfig1}
\end{figure}

We start with the isotropic case. A summary of our results
can be firstly found in Fig.~\ref{newfig1} which contains a two-parameter ($p$ vs $b$)
diagram. Here, the depicted curve separates the regime of spectrally
stable waves (under the curve) from spectrally and dynamically
unstable ones (over the curve) for fixed frequency $\omega=1$.
It is important to recall here that any pair $(b,\tilde{\omega}=1)$ for a given $b$ and
fixed $\tilde{\omega}$ can be converted upon rescaling to a pair $(\tilde{b}=1,\omega=1/b^2)$,
i.e., results pertaining to $b$ variation for fixed $\tilde{\omega}$ are
tantamount to ones with fixed $\tilde{b}$ and variable $\omega$. By using the latter representation, it is possible to connect to the well-known Vakhitov-Kolokolov criterion for the spectral stability, based on the monotonicity of the $P(\omega)$ dependence~\cite{promislow}. Increasing dependence of $P(\omega)$ (or, analogously, decreasing dependence of $P(b)$) is necessary for spectral stability, while a decreasing dependence (or increasing dependence of $P(b)$) leads to spectral (and dynamical) instability
for the single-humped states
considered herein. Furthermore, it should be noted
that the limit of $b=0$ is tantamount to $\omega \rightarrow \infty$, while $b \rightarrow 2$ corresponds to $\omega \rightarrow 1/4$ within the
above scaling (the linear limit), setting  the scales of variation of the respective parameters.

Representations of the dependence of $P$ with respect to $\omega$ for different values of $p$ can be found in Fig. \ref{newfig2}.
It can be clearly seen that in the case of $p=3$, similarly to what happens for all values with $1<p<3$, $P$ increases monotonically with $\omega$, pertaining to a regime of spectral stability.
Our numerical results seem to suggest the
presence of a $p^* \approx 3.3$
(see once again the right end of the
curves in the panels of Fig.~\ref{newfig1}).
For $1<p<p^*$, in line with Theorem 1,
we find ground state minimizers for {\it all}
values of $P \equiv \lambda$.
By $P$ here, we denote the squared $L^2$
norm due to its being tantamount to the
optical power in the corresponding
physical problem.
For values of the exponent $p$ that lie within $p^*<p<5$, the power $P$ features an interval of monotonic decrease with $\omega$ close to the linear limit (of dispersing waveforms). Indeed, the corresponding solutions near the linear limit are unstable, while for sufficiently large frequencies the solutions become spectrally stable, as seen in the top right panel which corresponds to $p=5$.
This finding also corroborates the results
of Theorem 1, since in the latter interval,
it is not possible to reach powers $P$
($\lambda$) below the minimum of the
corresponding curve.
For $p>5$ and below a critical, dimension-dependent threshold
(which for our two-dimensional case is $p_{cr} = 6.565$),
instabilities
arise {\it both} for sufficiently small (near the linear limit) and sufficiently large (instabilities due to collapse) values of $\omega$, as it is shown in the bottom left panel for $p=6$; in this case, the only stable frequencies are the intermediate ones, corresponding to the interval of growing $P$. Finally, when going above the relevant critical value of $p$ (see bottom right panel, corresponding to $p=7$), the soliton is spectrally and dynamically unstable for {\em all} the frequencies, given the monotonically decreasing dependence of $P$ with respect to $\omega$. Notice that these findings
for $p>5 \equiv 1+\frac{8}{d}$ complement
in a natural way
the rigorous results of Theorem 1.

\begin{figure}[htb]
    \centering
    \begin{tabular}{cc}
    \includegraphics[width=84mm,clip = true]{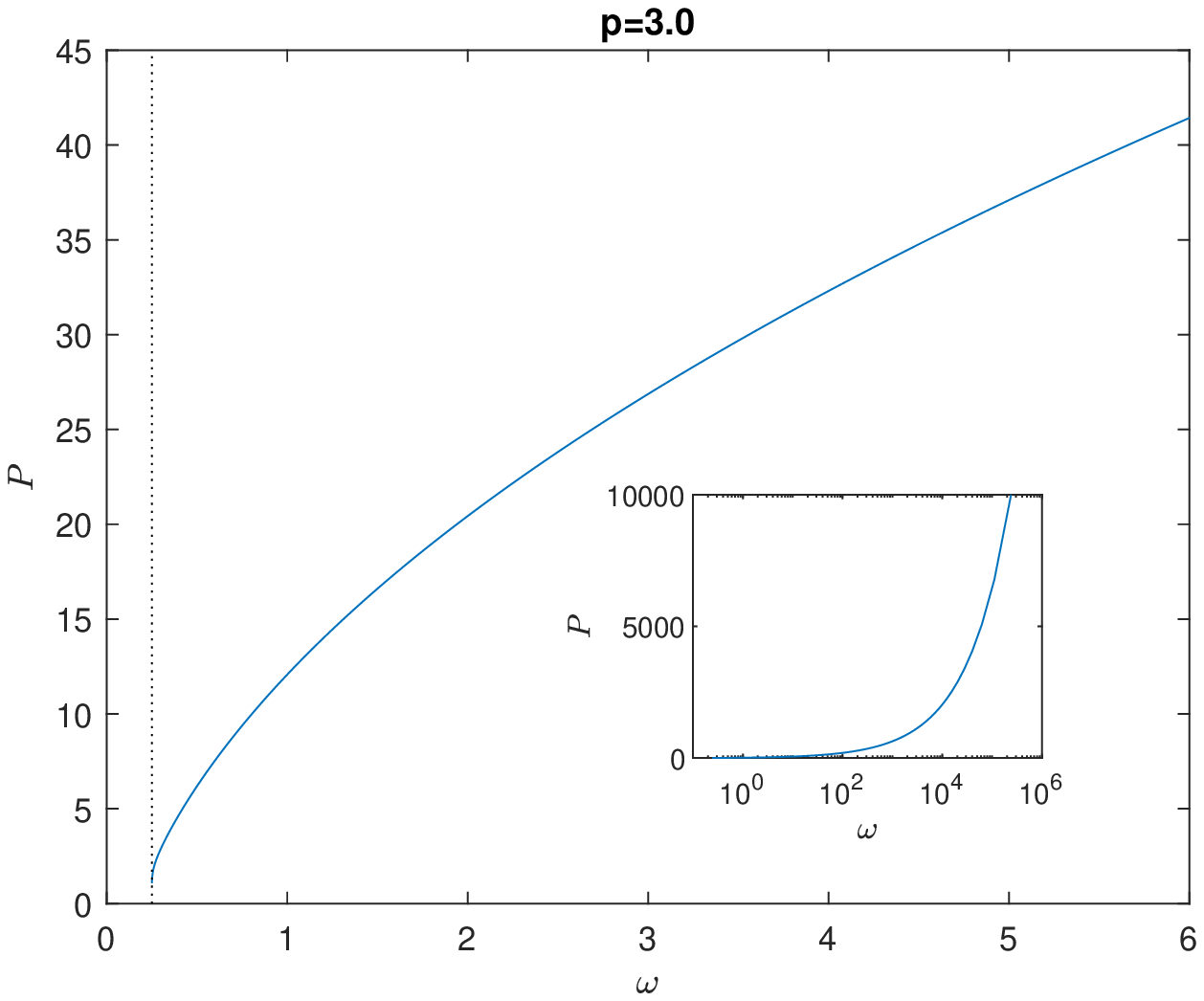} &
    \includegraphics[width=84mm,clip = true]{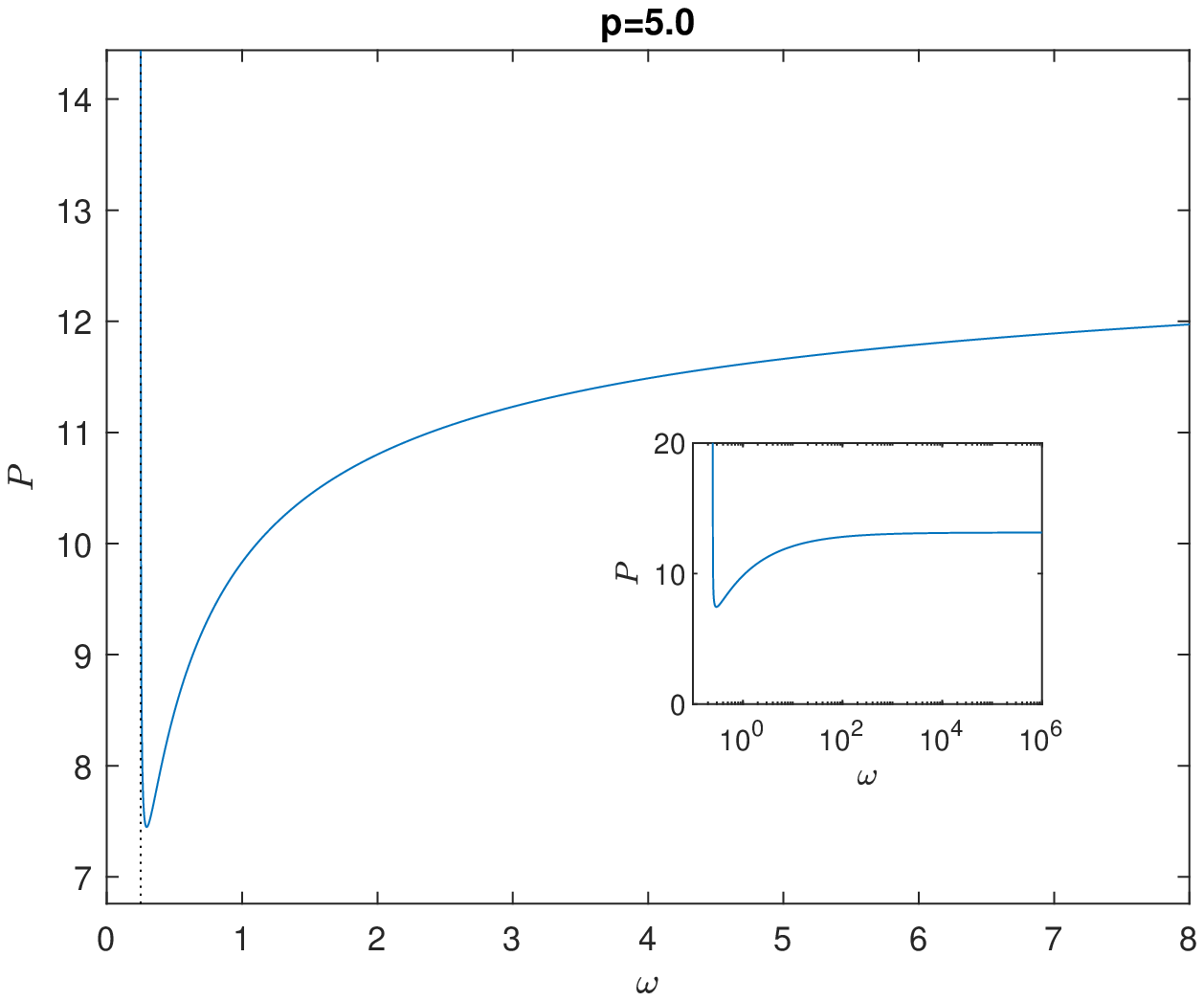} \\
    \includegraphics[width=84mm,clip = true]{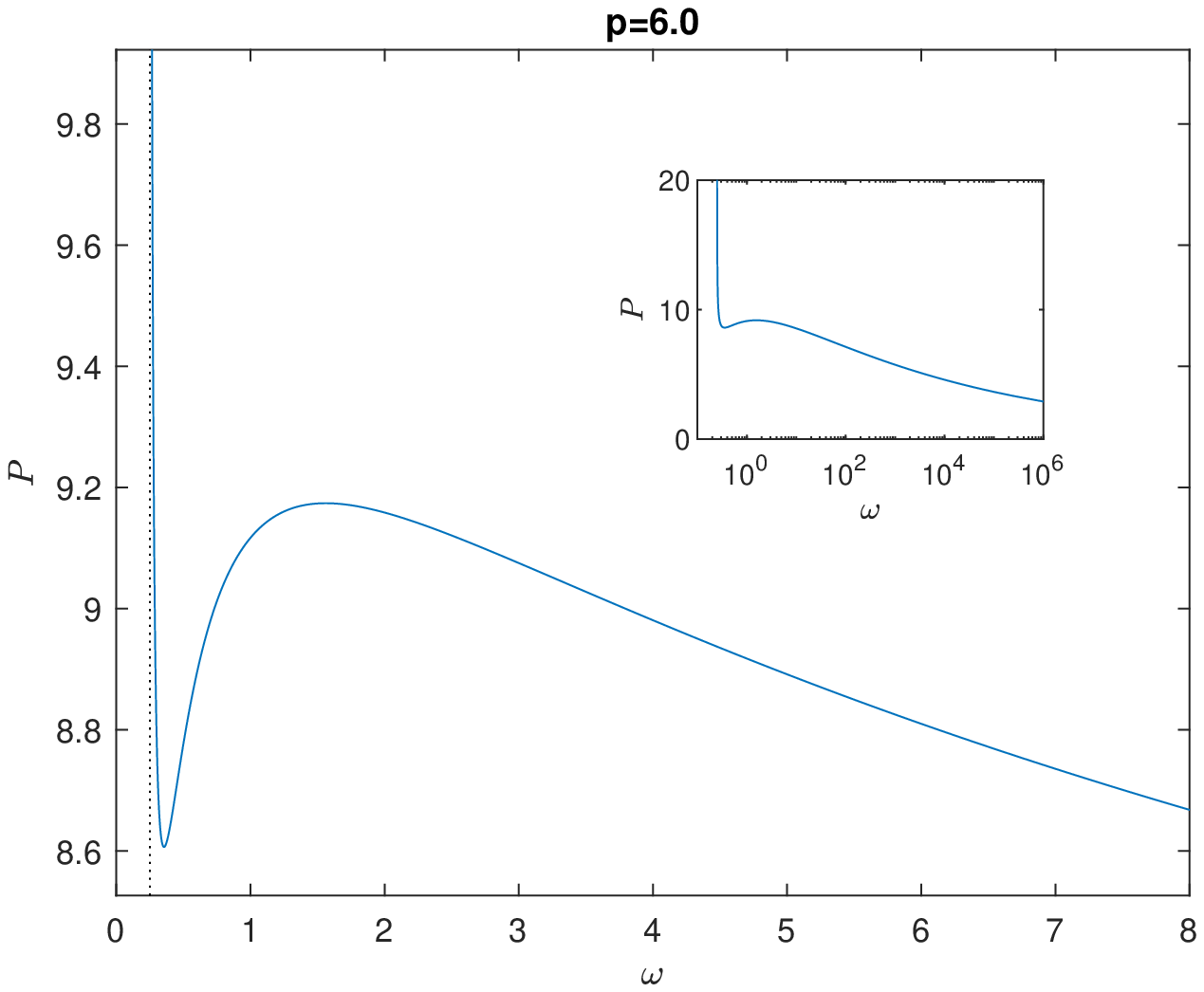} &
    \includegraphics[width=84mm,clip = true]{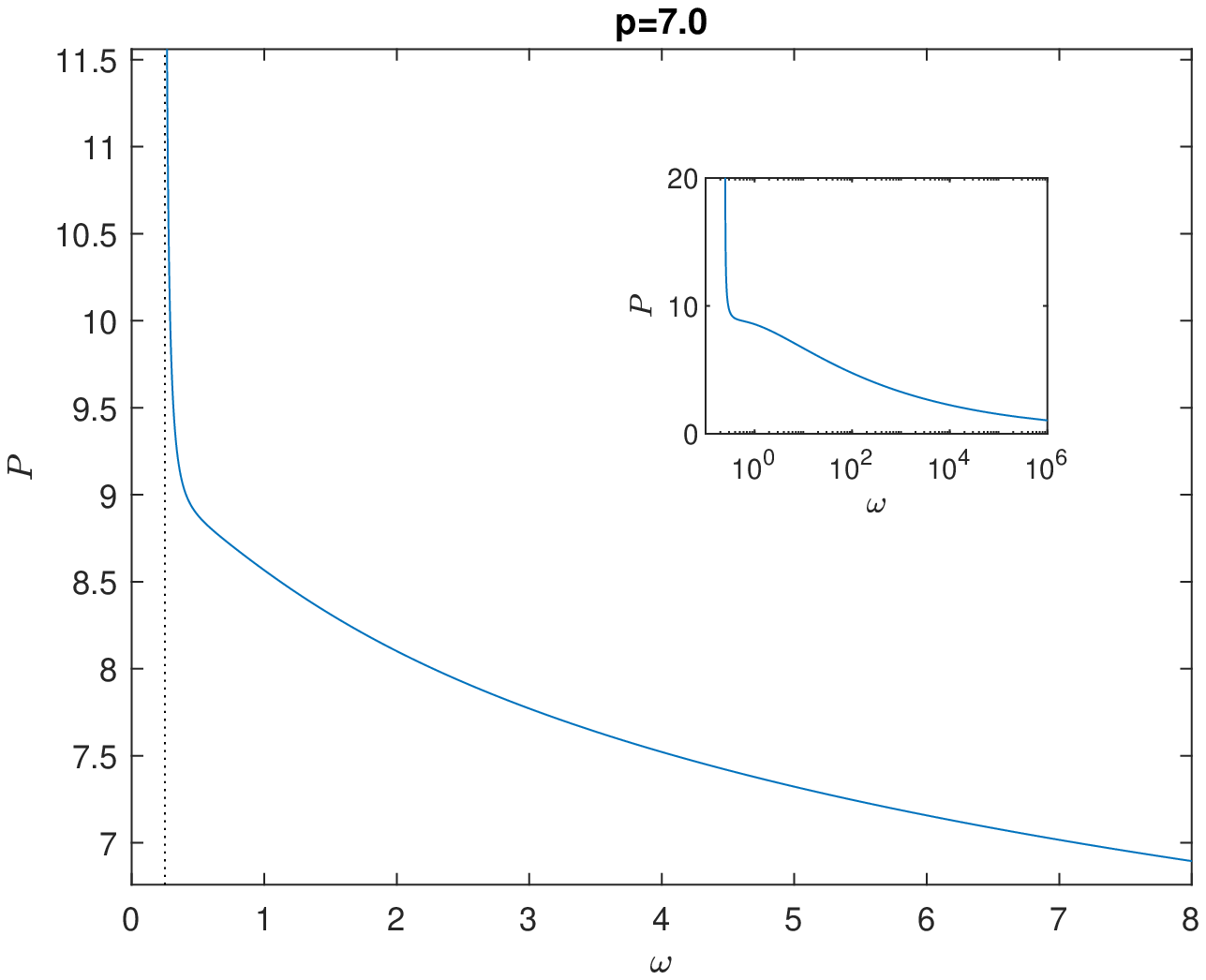} \\
    \end{tabular}
    \caption{Dependence of the squared $L^2$ norm, denoted by $P$,
    { i.e., $P=\int_{{\mathbf{R}^2}} |u|^{2}$ for our computations}, with respect to the frequency $\omega$ for different values of the nonlinearity exponent $p$, in the {\em isotropic} case for dimension $d=2$. These plots showcase the different stability regimes that can be found herein (see text for more details).
    The insets show the same graph over an expanded interval of
    frequencies, using a semi-logarithmic scale for the latter.}
    \label{newfig2}
\end{figure}

\begin{figure}[htb]
    \centering
    \includegraphics[width=150mm,clip = true]{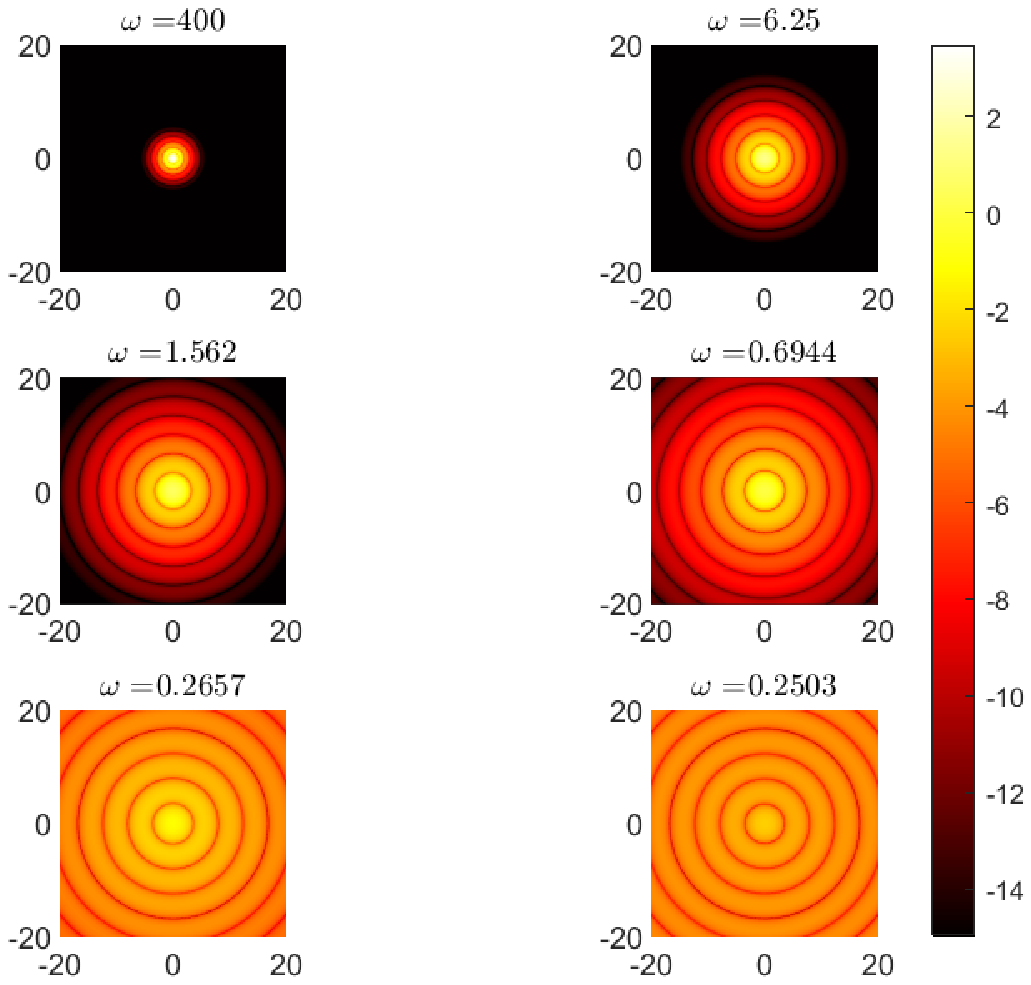}
    \caption{Several examples of the waveform of the solitary waves with $p=3$ in the isotropic case for different frequencies. We can observe how the solution profile changes from high $\omega$ to the linear limit of $\omega \rightarrow 0.25$. Notice the logarithmic scale of the colormap, and the (clear within that scale) zero-crossings of the solution. Figures for other values of $p$ are qualitatively similar.}
\label{newfig3}
\end{figure}

Figure \ref{newfig3} showcases the relevant isotropic (radially symmetric) waveforms and a variety of different frequencies, starting from the highly nonlinear limit of large $\omega$ (where the width of the solution shrinks,
while its amplitude grows), to progressively lower frequencies, eventually approaching the linear limit of small amplitude as $\omega \rightarrow 1/4$. It is important to note the logarithmic scale of the relevant colorbar, associated to continuously decreasing amplitudes as $\omega$ decreases. Noticeable also within this scale are the nodal lines of the solution, given the oscillatory nature of the linear tail as a result of the competition between the harmonic and biharmonic terms. Although this figure corresponds to $p=3$, it is qualitatively similar to the outcome for other values of $p$.

\begin{figure}[htb]
    \centering
    \includegraphics[width=80mm,clip = true]{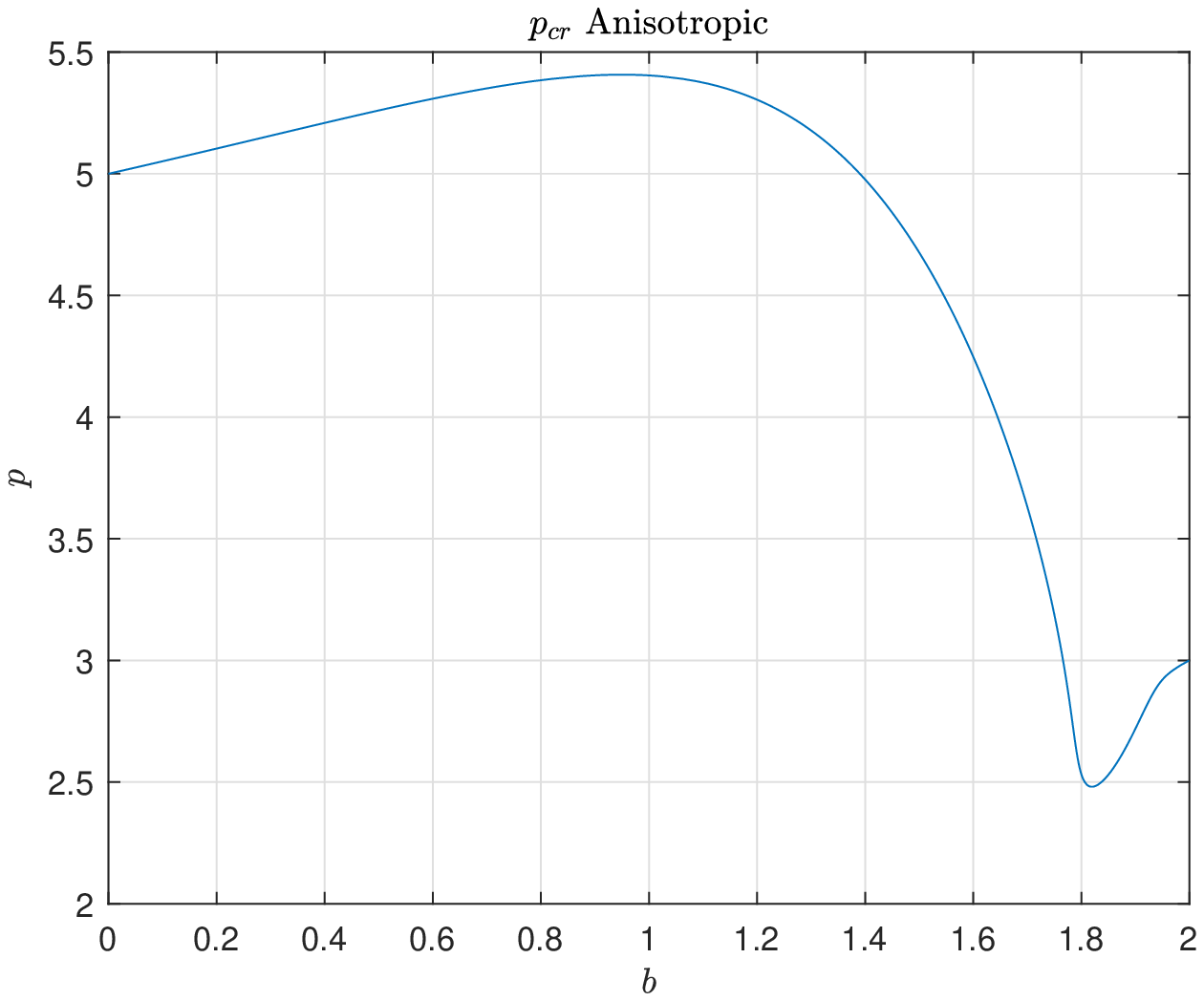}
    \caption{Same as Fig.~\ref{newfig1} but for the anisotropic case.}
    \label{newfig4}
\end{figure}

We have made a similar analysis for the anisotropic case. The two-parameter diagram of $p$ versus $b$ is displayed in Fig.~\ref{newfig4}. Indeed, the phenomenology is quite similar to the isotropic case, although with some notable differences that can be observed not only near the right edge
of the curve of Fig.~\ref{newfig4} but also
in the $P(\omega)$ plots for different values of $p$ in Fig.~\ref{newfig5}. For low enough $p$ (as, e.g., for $p=2$), the soliton is stable for every frequency,
and a solution exists for {\it all} values
of $P\equiv \lambda$, in line with Theorem 2.
However, contrary to the isotropic case, this
monotonic dependence of $P$ on $\omega$ does {\it not} persist up to $p=3$. Indeed, there exists an interval of $b$'s (or, equivalently,
of frequencies) for $p$ roughly between $2.481$ and $3$ in our two-dimensional setting, whereby $P(\omega)$ presents a maximum and a minimum and, as a consequence, the soliton becomes unstable in that interval (as shown, e.g., in the plot for $p=2.8$); this suggests that the linear limit is not approached in the same way as in the isotropic case near the critical point of $p=3$. Incidentally, it is especially relevant to note that the linear limit itself bears nontrivial differences as now the second partial derivative only occurs along $x$ direction. This leads, near the linear limit, to an oscillatory pattern {\it solely} along the $x$ direction, while the solution becomes separable in the form $X(x) \times Y(y)$. This can be clearly observed in the relevant solution panels in Fig.~\ref{newfig6}.

It is also relevant to note here that our
numerical results do not contradict
Theorem 2, although in the very vicinity
of the linear limit and for values of
$p$ between $2.5$ and $3$, we cannot
fully confirm the relevant theory.
In particular, a careful observation
of Fig.~\ref{newfig4}, e.g., for $p=2.8$
(top right panel) suggests a non-monotonic
dependence of $P$ on $\omega$ but as
the linear limit is approached, we are
unable to resolve the question of whether
{\it all} values of $P$ are accessible,
as one approaches closer and closer to
$\omega=1/4$, in line with the expectations
of Theorem 2. While the Theorem prompts
us to expect that to be the case (and
the numerics are also suggestive in this vein), the
highly computationally expensive,
anisotropic
2d computations needed have not allowed us
to fully confirm this limit, which remains an
interesting, open
computational question for future
studies.

\begin{figure}[htb]
    \centering
    \begin{tabular}{cc}
    \includegraphics[width=84mm,clip = true]{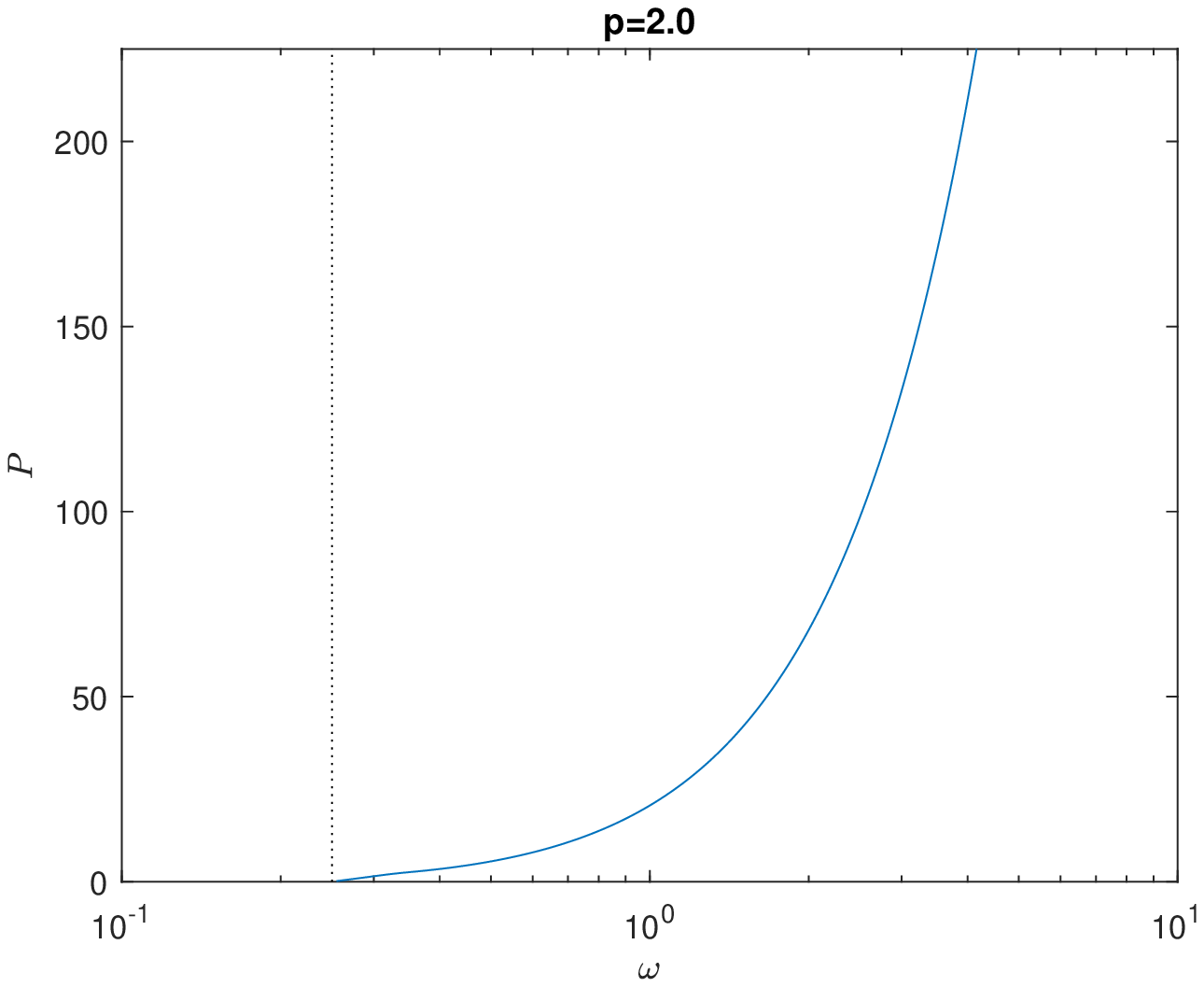} &
    \includegraphics[width=84mm,clip = true]{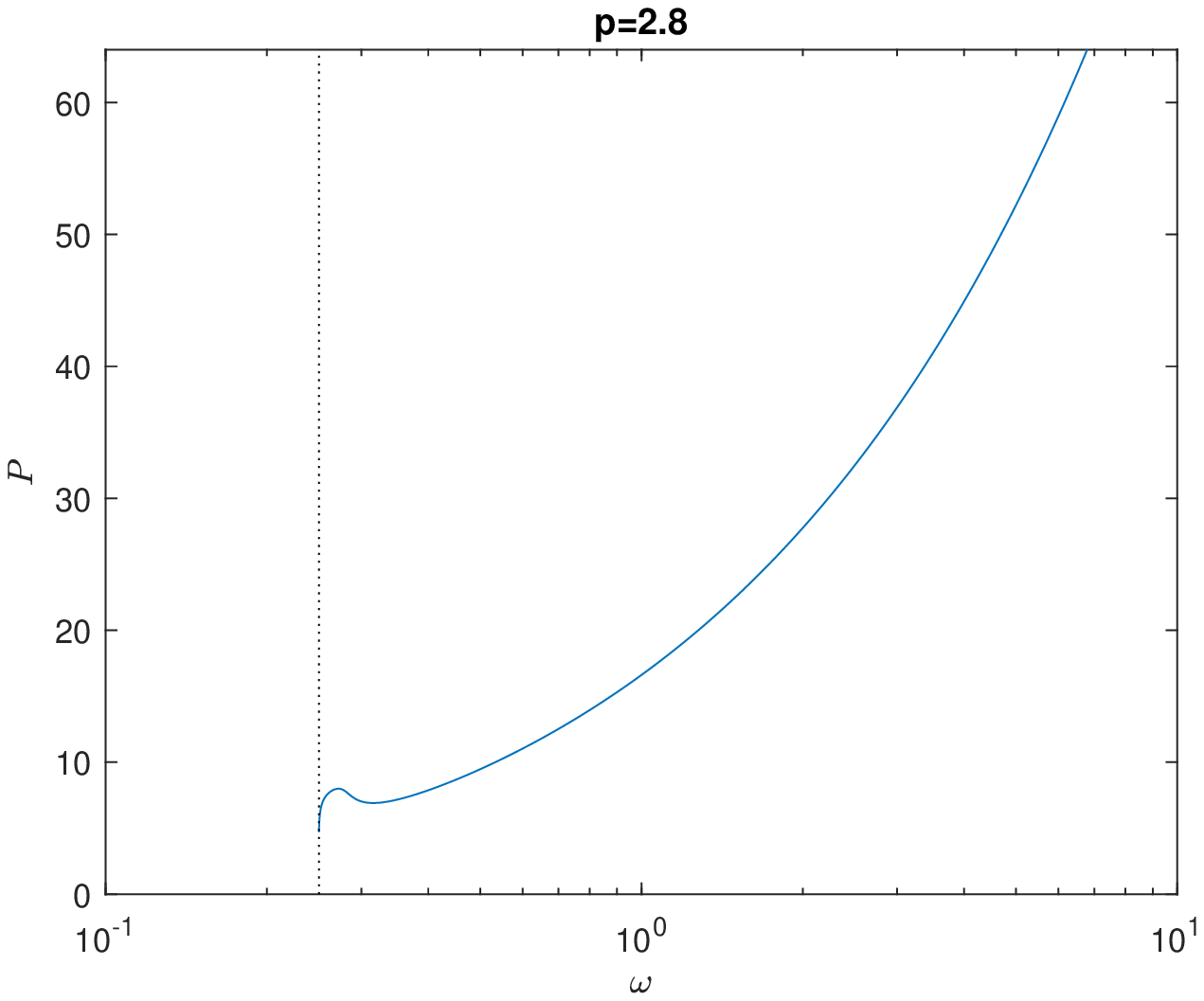} \\
    \includegraphics[width=84mm,clip = true]{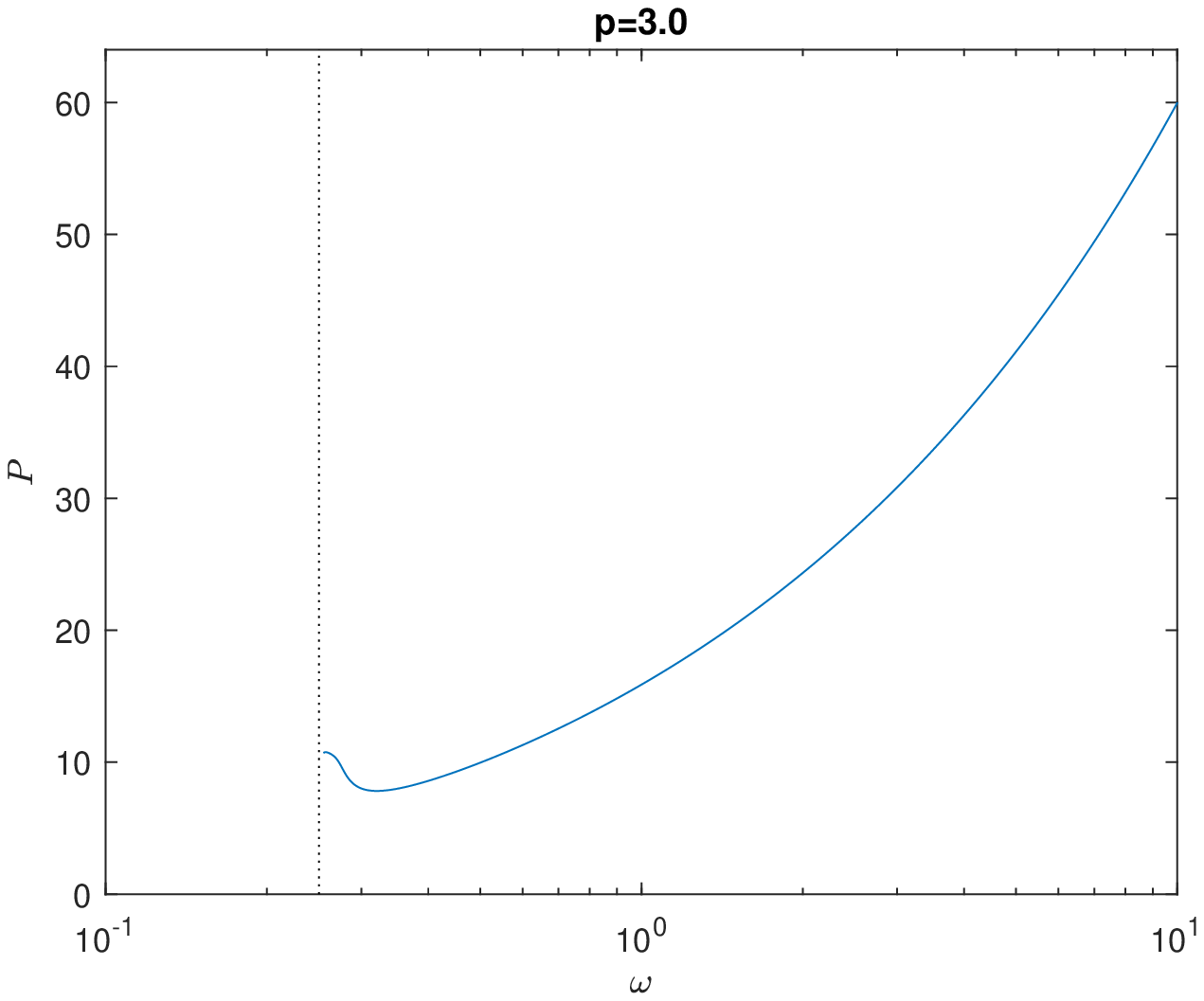} &
    \includegraphics[width=84mm,clip = true]{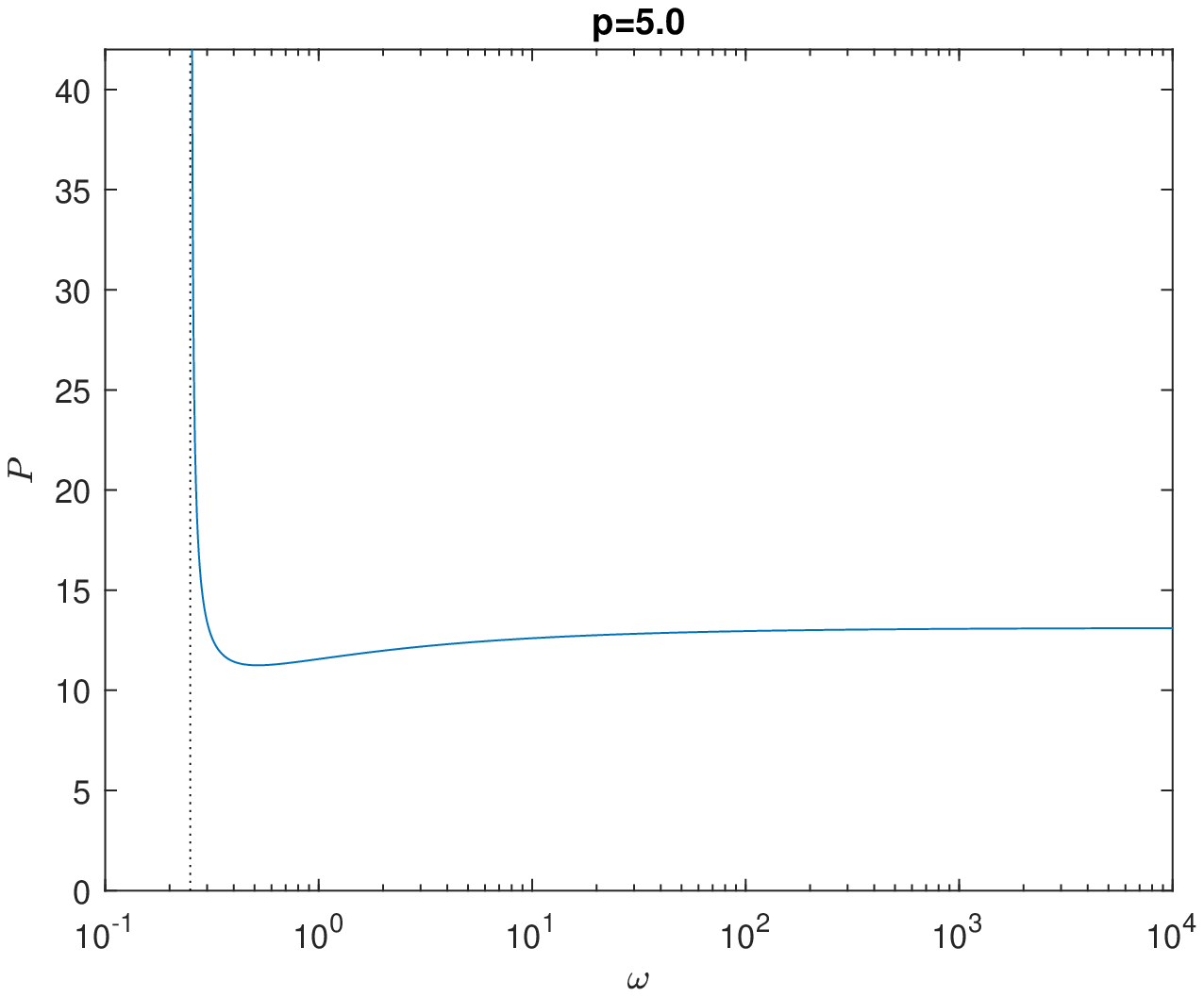} \\
    \includegraphics[width=84mm,clip = true]{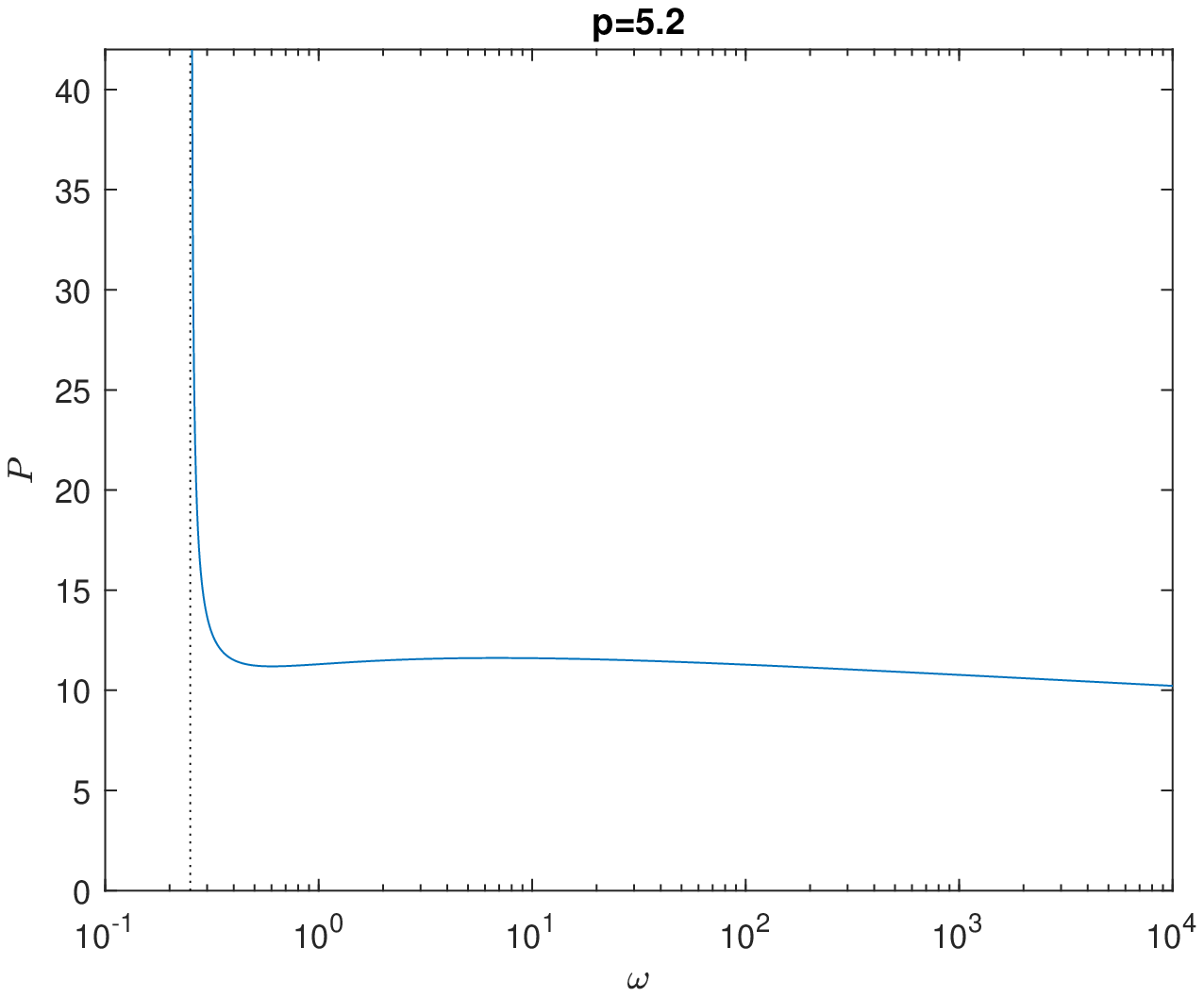} &
    \includegraphics[width=84mm,clip = true]{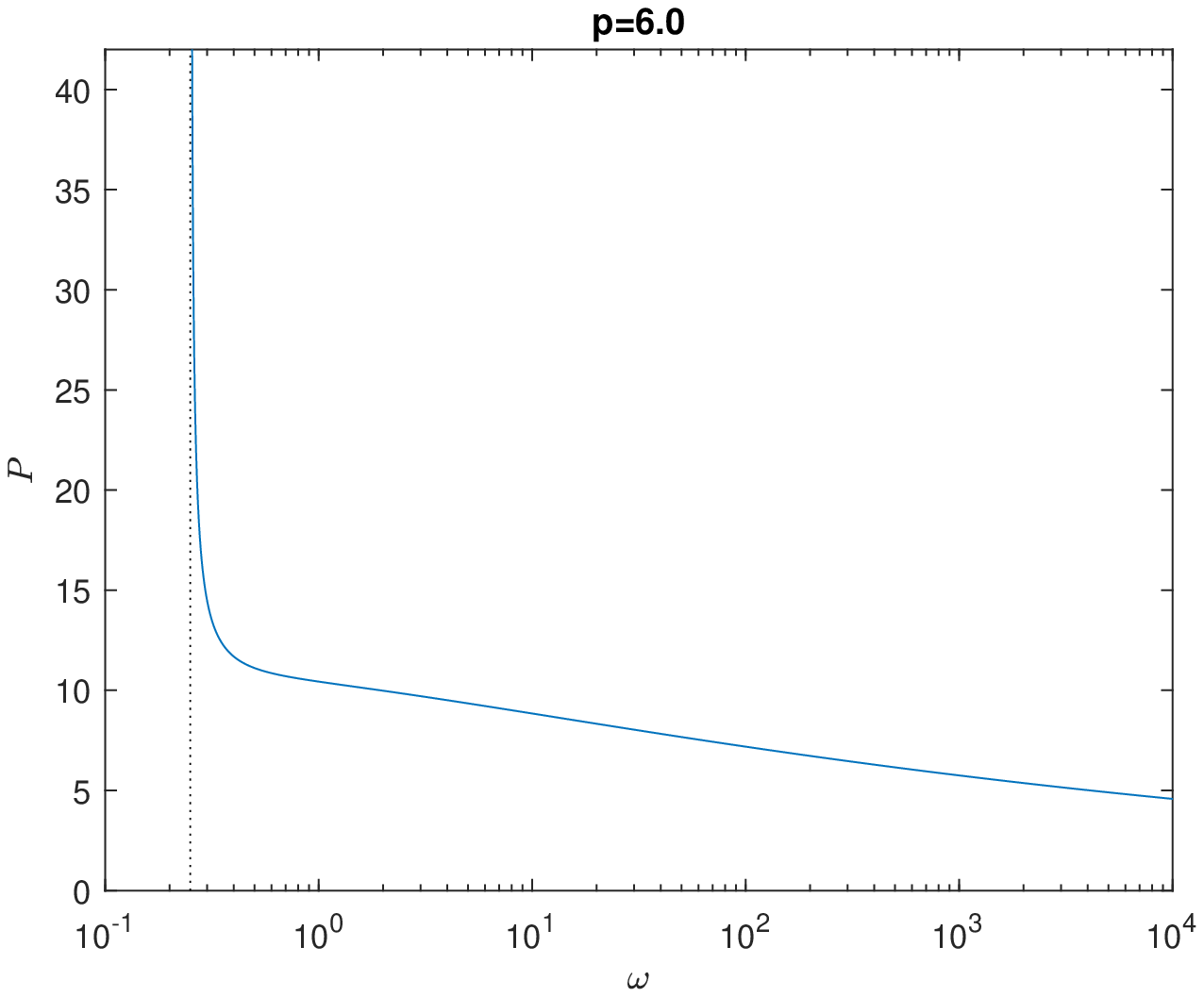}
    \end{tabular}
     \caption{Same as Fig.~\ref{newfig2} but for the anisotropic case and for
     different values of $p$. Again, a
     semi-logarithmic scale has been used for the frequencies.}
     \label{newfig5}
\end{figure}

When $p$ is increased from $p=3$, we observe a similar phenomenology as in the isotropic case, i.e. the curve $P(\omega)$ is monotonically decreasing near the linear limit and becomes monotonically increasing (spectrally stable) after a local minimum (see the plot for $p=5$). This phenomenology changes again (resembling the isotropic case) for $p>5$, as shown in the plot for $p=5.2$; here, an interval of stability for intermediate frequencies can be seen to arise. Finally, for sufficiently large values of $p$,
again similarly to the isotropic limit, the waves
become generically unstable for all frequencies, as illustrated by the monotonically decreasing dependence of $P(\omega)$ in the plot for $p=6$. It is interesting to point out, however, that the relevant threshold is considerably lower in the anisotropic case where it is around $p=5.407$ for $d=2$, while in the isotropic one the threshold is around $p=6.565$ for $d=2$.

\begin{figure}[htb]
    \centering
    \includegraphics[width=150mm,clip = true]{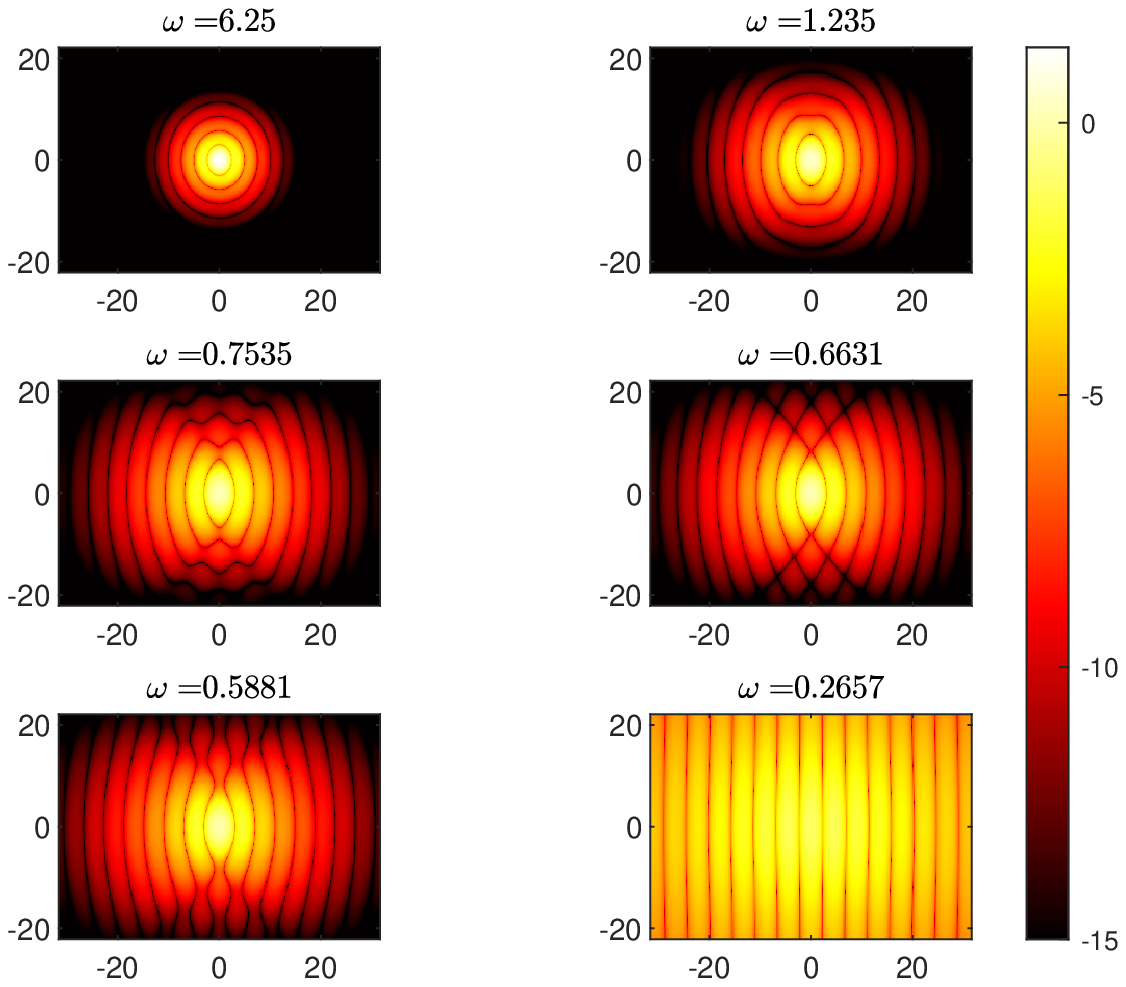}
    \caption{Same as Fig.~\ref{newfig3} but for the anisotropic case with $b=1$. Contrary to the isotropic case, the anisotropy reflects in the solution as it acquires, when approaching the linear limit $\omega\rightarrow0.25$ a separable form in the $x$ and $y$ dependence with the nodal lines being uniform along direction $y$.}
\label{newfig6}
\end{figure}

\section{Conclusions and Future Challenges}

In the present work, we have examined
in a systematic way the properties of
higher-dimensional NLS models
with mixed dispersion with a numerical emphasis
on the more computationally tractable case of $d=2$. In
particular, we have considered a setting
in which there is
a competition between a focusing quartic
and a defocusing quadratic dispersion term.
Our Theorems 1 and 2 have offered a rigorous
perspective on the relevant phenomenology,
providing bounds on the nonlinearity exponent
(as a function of dimension) for which
minimizers of the (squared) $L^2$ norm
exist for all values of that quantity,
as well as ones for which such minimizers
do not exist for all powers. This has been
done both for the isotropic case involving
radial solutions, as well as for the
anisotropic one where the second derivative
term was only active along a particular
direction.
We have complemented these findings with
detailed numerical results and corresponding
multi-parameter diagrams detailing the stability
of the single-humped states of the system.
In both isotropic and anisotropic cases,
we found that the exponent $p$ of the nonlinear term
is sufficiently small, the dependence of
the power on the frequency is monotonic,
while above a certain threshold more
complex non-monotonic dependencies arise.
Our numerical results for the case of $d=2$ go beyond the
accessible limits to our Theorems,
identifying possible stable solutions
even beyond the exponent bound of $p=1+8/d$
(where $d$ is the dimension), as well
as identifying exponents beyond which no
spectrally stable solutions arise.

While we believe that these results
offer numerous insights into higher-dimensional
systems with competing quadratic and quartic
dispersions, there are also numerous open
questions to consider. As concerns the
problem examined herein, an important one
concerns the close proximity of the linear
limit for the anisotropic case when $p$ is
in the vicinity of $p=3 \equiv 1+4/d$
(for our case of $d=2$). More generally,
for the case considered herein, it would
be interesting to also examine if higher-excited
states, including multi-soliton ones, as
well as vortical ones are feasible and
potentially also spectrally stable
(and under what conditions).
Additionally, numerical studies of the more computationally
demanding case of $d=3$ would also be worthwhile
in connection to the Theorems presented herein.
Last but not least, exploring similar questions
with the recently accessible experimentally,
even orders of dispersion~\cite{RungePRR2021}
would also be of particular interest.
Such studies are presently under
consideration and will be reported in
future publications.

\section*{Acknowledgment}
A.S. acknowledges partial support from NSF-DMS \# 1908626 and \# 2204788. J.C.-M. acknowledges support from EU (FEDER program 2014-2020) through both Consejería de Economía, Conocimiento, Empresas y Universidad de la Junta de Andalucía (under the projects P18-RT-3480 and US-1380977), and MICINN and AEI (under the projects PID2019-110430GB-C21 and PID2020-112620GB-I00.

\end{document}